\documentclass[a4paper,english,pdfsync]{lipics-v2018}

\nolinenumbers

\usepackage{hyperref}
\usepackage{microtype}
\usepackage[utf8]{inputenc}
\usepackage{thm-restate}
\usepackage{bbm}
\usepackage{mathabx}
\usepackage{tikz} 
\usetikzlibrary{patterns}
\usepackage[textwidth=2cm,textsize=small]{todonotes}
\usepackage{enumerate}
\usepackage{xspace}
\usepackage{paralist}
\usepackage{calligra}
\usepackage{wrapfig}
\usepackage{extpfeil}

\newcommand{\ignore}[1]{}

\newcommand{\goesto}[1]{\xrightarrow{#1}}

\newcommand{\A}{\mathcal A}
\newcommand{\T}{\mathcal T}
\renewcommand{\S}{\mathcal S}

\newcommand{\lts}{\mbox{\sc lts}\xspace}
\newcommand{\stca}{\mbox{\sc tca}\xspace}
\newcommand{\tca}{\mbox{\sc tca}\xspace}
\newcommand{\ca}{\mbox{\sc ca}\xspace}
\newcommand{\rac}{\mbox{\sc rac}\xspace}

\newcommand{\ta}{\mbox{\sc ta}\xspace}

\newcommand{\true}{\mathbf{true}}

\newcommand{\N}{\mathbb N}
\newcommand{\Z}{\mathbb Z}
\newcommand{\Q}{\mathbb Q}
\newcommand{\Qpos}{\Q_{\geq 0}}
\newcommand{\I}{\mathbb I}
\newcommand{\C}{\mathtt C}

\renewcommand{\P}{\mathtt P}
\newcommand{\R}{\mathcal R}
\newcommand{\Reg}{\mathtt R}
\renewcommand{\SS}{\mathtt S}
\newcommand{\SSyes}{\SS_{\textrm{yes}}}
\newcommand{\SSno}{\SS_{\textrm{no}}}
\newcommand{\Xyes}{\X^\p_{\textrm{yes}}}
\newcommand{\Cntyes}{\Cnt_{\textrm{yes}}}
\newcommand{\n}{\mathtt n}
\newcommand{\Cnt}{\mathtt N}
\newcommand{\K}{\mathcal K}
\newcommand{\M}{\mathtt M}
\newcommand{\m}{\mathtt m}
\newcommand{\X}{\mathtt X}
\newcommand{\Y}{\mathtt Y}
\newcommand{\p}{\mathtt p}
\newcommand{\q}{\mathtt q}
\renewcommand{\r}{\mathtt r}
\newcommand{\s}{\mathtt s}
\newcommand{\pq}{\mathtt {pq}}
\newcommand{\qp}{\mathtt {qp}}
\newcommand{\qr}{\mathtt {qr}}

\newcommand{\x}{\mathtt x}
\newcommand{\y}{\mathtt y}
\newcommand{\z}{\mathtt z}
\renewcommand{\c}{\mathtt c}
\renewcommand{\L}{\mathtt L}
\renewcommand{\l}{\mathtt l}
\newcommand{\op}{\mathsf{op}}
\newcommand{\Op}{\mathsf{Op}}
\newcommand{\nop}{\mathsf{nop}}
\newcommand{\elapse}{\mathsf{elapse}}
\newcommand{\reset}{\mathsf{reset}}
\newcommand{\resetop}[1]{{\reset(#1)}}
\newcommand{\sndop}[3]{{\mathsf{send}(#2, #1 : #3)}}
\newcommand{\rcvop}[3]{{\mathsf{receive}(#2, #1 : #3)}}
\newcommand{\copyc}[2]{\psi_{\mathsf{copy}}^{#1#2}}

\newcommand{\fract}[1]{\left\{#1\right\}}
\newcommand{\incrop}[1]{#1\texttt{++}}
\newcommand{\decrop}[1]{#1\texttt -\texttt-}
\newcommand{\guessop}[1]{{\mathsf{guess}(#1)}}
\newcommand{\test}{\mathsf{test}}
\newcommand{\testop}[1]{{\test(#1)}}

\newcommand{\restrict}[2]{\left.#1\right|_{#2}}
\renewcommand{\vec}[1]{\overline{#1}}
\newcommand{\eqv}[1]{\equiv_{#1}}

\newcommand{\condone}[1]{\mathbbm 1_{#1?}}

\newcommand{\psimod}{\psi_{\mathsf{mod}}}
\newcommand{\orbit}[1]{\mathrm{O}(#1)}

\newcommand*{\threesim}{%
      \mathrel{\vcenter{\offinterlineskip
      \hbox{$\sim$}\vskip-.35ex\hbox{$\sim$}\vskip-.35ex\hbox{$\sim$}}}}

\newcommand{\sem}[1]{\left\ldbrack#1\right\rdbrack}
\newcommand{\desem}[1]{\left\ldbrack#1\right\rdbrack^\textsf{de}}
\newcommand{\desynchto}[1]{\goesto{#1}^\textsf{de}}
\newcommand{\rvsem}[1]{\left\ldbrack#1\right\rdbrack^\textsf{rv}}
\newcommand{\rvto}[1]{\goesto{#1}\mbox{}\!\!^\textsf{rv}}

\newcommand{\floor}[1]{\lfloor#1\rfloor}

\newcommand{\tuple}[1]{\left\langle#1\right\rangle}

\newcommand{\family}[3]{(#1^{#2})_{#2 \in #3}}

\newcommand{\set}[1]{\left\{ #1 \right\}}
\newcommand{\setof}[2]{\set{#1 \; \middle| \; #2}}

\newif\ifstartedinmathmode
\renewcommand*{\st}{
  \relax\ifmmode\startedinmathmodetrue\else\startedinmathmodefalse\fi
  \ifstartedinmathmode{\;\cdot\;}\else{s.t.~}\fi%
}

\newcommand{\ie}{i.e.\xspace}
\newcommand{\wlg}{w.l.o.g.~}
\newcommand{\Wlg}{W.l.o.g.~}

\newcommand{\wrt}{w.r.t.~}
\newcommand{\cf}{cf.~}

\newcommand{\vs}{vs.~}
\newcommand{\resp}{resp.\xspace}

\newcommand{\stkout}[1]{\ifmmode\text{\sout{\ensuremath{#1}}}\else\sout{#1}\fi}

\DeclareMathAlphabet{\mathcalligra}{T1}{calligra}{m}{n}
\DeclareFontShape{T1}{calligra}{m}{n}{<->s*[2.2]callig15}{}
 
\newcommand{\scriptr}[1]{\ensuremath{\mathcalligra{#1}}}
\newcommand{\arr}{\scriptr r}

\newcommand{\myitem}[1]{\noindent \textbf{(#1)}}
\newcommand{\interitem}[2]{\intertext{\myitem {#1} {#2}}}

\newcommand{\mysubparagraph}[1]{\vspace{-1em}\subparagraph{#1}}

\title{Decidability of Timed Communicating Automata}

\author{Lorenzo Clemente}{University of Warsaw\\{Warsaw, Poland}}{clementelorenzo@gmail.com}{0000-0003-0578-9103}{}

\authorrunning{L. Clemente}

\Copyright{Lorenzo Clemente}

\subjclass{\ccsdesc[500]{Theory of computation~Distributed computing models}.
\ccsdesc[500]{Theory of computation~Timed and hybrid models}.
\ccsdesc[300]{Theory of computation~Logic}}

\keywords{timed automata, communicating automata, reachability problem, quantifier elimination, register automata}

\category{}

\relatedversion{}

\supplement{}

\funding{}

\acknowledgements{}

\begin{document}

\maketitle

\begin{abstract}
	We study the reachability problem for networks of timed communicating processes.
	Each process is a timed automaton 
	communicating with other processes by exchanging messages over unbounded FIFO channels.
	Messages carry clocks which are checked at the time of transmission and reception with suitable timing constraints.
	Each automaton can only access its set of local clocks and message clocks of sent/received messages.
	Time is dense and all clocks evolve at the same rate.
	Our main contribution is a complete characterisation of decidable and undecidable communication topologies
	generalising and unifying previous work.
	From a technical point of view,
	we use quantifier elimination and a reduction to counter automata with registers.
\end{abstract}


\section{Introduction}
\label{sec:introduction}

\emph{Timed automata} (\ta) were introduced almost thirty years ago by Alur and Dill \cite{AlurDill:ICALP:1990,AlurDill:NTA:TCS:1994}
as a decidable model of real-time systems elegantly combining finite automata with timing constraints over a densely timed domain.
To these days, \ta are still an extremely active research area,
as testified by recent works on topics such as
the reachability problem \cite{HerbreteauSrivathsanWalukiewicz:IC:2016},
a novel analysis technique based on tree automata 
\cite{AkshayGastinKrishnaSarkar:CONCUR:2017},
and 
the binary reachability relation \cite{QuaasShirmohammadiWorrell:LICS:2017}.
Decidability results on \ta have been extended to include discrete data structures
such as counters \cite{TimedPetriNets:05,AbdullaMahataMayr:DTPNs:LMCS2007},
stacks \cite{bouajjani:timed:PDA:94,Dang:PTA:2003,TrivediWojtczak:RTA:2010,BenerecettiMinopoliPeron:RTA:2010,AbdullaAtigStenman:DensePDA:12,Quaas:LMCS:2015,ClementeLasota:LICS:2015,ClementeLasotaLazicMazowiecki:LICS:2017,ClementeLasota:ICALP:2018},
and lossy FIFO channels \cite{AbdullaAtigCederberg:TLCS:FSTTCS12};
\cf the recent survey \cite{WaezDingelRudie:CSR:2013} for more examples of \ta extensions.

In this paper, we study \emph{systems of timed communicating automata} (\stca) \cite{KrcalYi:CTA:CAV:2006},
which are networks of \ta processes
exchanging messages over FIFO channels (queues) of unbounded size%
\footnote{The original name \emph{communicating timed automata} \cite{KrcalYi:CTA:CAV:2006} refers to a version of \stca with untimed channels.
In order to stress that we consider timed channels, we speak about timed communicating automata.}.
Messages are additionally equipped with densely-valued clocks which elapse at the same rate as local \ta clocks.
When a message is sent, a logical constraint between local and message clocks specifies the initial values for the latter;
if multiple values are allowed, a satisfying one is chosen nondeterministically.
Symmetrically, when a message is received, a logical constraint on local and message clocks
specifies whether the reception is possible.
%

We consider three kinds of clocks:
\emph{classical clocks} over the rationals $\Q$,
\emph{integral clocks} over the nonnegative integers $\N$,
and \emph{fractional clocks} over the unit interval $\I := \Q \cap [0, 1)$.
All clocks evolve at the same rate;
an integral clock behaves the same as a classical clock,
except that in constraints it evaluates to the underlying integral part;
when a fractional clock reaches value $1$,
its value is wrapped around $0$.
Integral and fractional clocks are complementary in the sense that they express two perpendicular features of time:
Integral clocks are unbounded but discrete,
and fractional clocks are bounded but dense.
For classical and integral clocks $\x, \y$,
we consider \emph{inequality} $\x - \y \sim k$ and \emph{modulo} $\x - \y \eqv m k$ constraints;
for fractional clocks $\x, \y : \I$ we consider \emph{order} constraints $\x \sim \y$,
where $\sim \in \set {<, \leq, \geq, >}$.
In the presence of fractional clocks, constraints on classical and integral clocks are inter-reducible.
Nevertheless, we consider separately classical, integral, and fractional clocks, mainly for two reasons.
First, in our main result below we can point out with greater precision what makes the model computationally harder.
Second, from a technical standpoint
it is sometimes more convenient to manipulate classical clocks%
---their constraints are invariant \wrt the elapse of time;
sometimes integral clocks%
---they reduce the impedance when converting to counters.
%

The \emph{non-emptiness problem} asks whether there exists an execution of the \stca
where all processes start and end in predefined control locations,
with empty channels both at the beginning and at the end of the execution.
It is long-known that already in the untimed setting of communicating automata (\ca)
the model is Turing-powerful \cite{BrandZafiropulo:CFM:ACM:1983},
and thus all verification questions such as non-emptiness are undecidable.
Decidability can be regained by restricting the \emph{communication topology},
i.e., the graph where vertices are processes $\p$,
and there is an edge $\p \to \q$
whenever there is a channel from process $\p$ to process $\q$.
A \emph{polytree} is a topology whose underlying undirected graph is a tree;
%
a \emph{polyforest} is a disjoint union of polytrees.
Our main result is a complete characterisation of the decidable \stca topologies.

\begin{restatable*}{theorem}{thmmain}
	\label{thm:main}
	Non-emptiness of \stca is decidable if, and only if,
	the communication topology is a polyforest
	\st in each polytree there is at most one channel with inequality tests.
\end{restatable*}
\noindent
Notice that fractional clocks do not influence decidability,
as neither do modulo constraints;
the characterisation depends only on which polytrees contain inequality tests,
on classical or integer clocks.
\ignore{
A similar characterisation was previously shown for \stca with \emph{untimed channels},
i.e., channels where messages do not carry any additional clocks.
More precisely, when all channels are untimed,
it was shown in \cite[Theorem 5]{ClementeHerbreteauStainerSutre:FOSSACS13}
that non-emptiness is decidable if, and only if, the topology is a polyforest.
Since an untimed channel in particular does not have integral inequality tests,
our results subsumes \cite{ClementeHerbreteauStainerSutre:FOSSACS13}.

If we additionally allow \emph{emptiness} channel tests, or, equivalently \emph{urgent receptions},
\emph{over discrete time} emptiness is decidable if, and only if,
the topology is a polyforest
and moreover in each polytree there is at most one channel that can be tested for emptiness
\cite[Theorem 3]{ClementeHerbreteauStainerSutre:FOSSACS13}.
Since integral inequality tests on message clocks can simulate emptiness channel tests and urgent receptions (\cf Sec.~\ref{sec:result}),
and since dense time subsumes discrete time,
our characterisation subsumes this result as well.
}
\noindent
This subsumes recent analogous characterisations
for \stca with untimed channels
in discrete \cite[Theorem 3]{ClementeHerbreteauStainerSutre:FOSSACS13} and dense time \cite[Theorem 5]{ClementeHerbreteauStainerSutre:FOSSACS13}.
It is worth remarking that we consider timed channels,
which were not previously considered
with the exception of the work \cite{AbdullaAtigKrishna:Arxiv:2017},
which however discussed only discrete time.
More precisely, it was shown there that,
with (integral) non-diagonal inequality tests of the form $\x \sim k$,
the topology $\p \to \q$ is decidable \cite[Theorem 4]{AbdullaAtigKrishna:Arxiv:2017},
while $\p \to \q \to \r$ is undecidable \cite[Theorem 3]{AbdullaAtigKrishna:Arxiv:2017}.
Since our undecidability result holds already in discrete time,
it follows from Theorem~\ref{thm:main} that $\p \to \q \to \r$ is undecidable;
additionally, new undecidable topologies can be deduced,
such as $\p \to_1 \q \to_2 \r \to_3 \s$
with $\to_1, \to_3$ with integral inequality tests and $\to_2$ untimed.

Regarding decidability,
Theorem~\ref{thm:main} vastly generalises all the previously known decidability results,
since it considers the more challenging case of timed channels,
it includes more topologies,
a richer set of clocks comprising both classical, integral, and fractional clocks,
a richer set of constraints comprising both diagonal and non-diagonal constraints,
and the more general setting of dense time.
In particular, combining timed channels with diagonal constraints on message and local clocks
was not previously considered.
Our characterisation completes the picture of decidable \stca topologies in dense time.

\mysubparagraph{Technical contribution.}

While our undecidability results are essentially inherited from \cite{ClementeHerbreteauStainerSutre:FOSSACS13},
the novelty of our approach consists in two main technical contributions of potentially independent interest,
which are used to show decidability.
First, we show that 
diagonal channel constraints reduce to non-diagonal ones
by the method of \emph{quantifier elimination};
\cf Lemma~\ref{lemma:copy-send} in Sec.~\ref{sec:simple}. 
This is a novel technique in the study of timed models
%
and we believe that its application to the study of timed models has independent interest,
as recently shown in the analysis of timed pushdown automata \cite{ClementeLasota:ICALP:2018}.

\ignore{
We then use two standard techniques in the analysis of \stca.
First, we consider a \emph{desynchronised semantics}
where we allow receiver processes to move ahead of time \wrt sender processes
(sound over-approximation; no constraint on the topology).
Second, we consider a \emph{rendezvous semantics}
where we force transmissions and corresponding receptions to happen simultaneously
(sound under-approximation; only for polyforest topologies).
These two semantics taken together allow us to remove channels over polyforest topologies
at the cost of introducing non-negative counters
measuring the integer desynchronisation between sender and receive processes
(\cf \cite{KrcalYi:CTA:CAV:2006,ClementeHerbreteauStainerSutre:FOSSACS13,AbdullaAtigKrishna:Arxiv:2017});
if a channel has integral inequality tests, then the corresponding counter has zero tests.
}

Our second technical contribution is the encoding of fractional clocks into $\I$-valued registers
over the \emph{cyclic order} $K \subseteq \I^3$,
\ie, the ternary relation $K(a, b, c)$ that holds whenever going clockwise on the unit circle starting at $a$,
we first visit $b$, and then $c$.
Cyclic order provides the most suitable structure to handle fractional values
and simplifies the technical development.
We believe this has wider application to the analysis of timed systems.
%
%
%
%
%
%

With the two technical tools above in hand,
for a given \stca over a polyforest topology
we build an equivalent \emph{register automaton with counters} (\rac) of exponential size. 
We establish decidability of non-emptiness for \rac by reducing to finite automata with counters.
If every polytree has at most one channel with integral inequality tests,
then one zero tests suffices,
and the latter model is decidable \cite{Reinhardt:TCS:2008,Bonnet:MFCS:2011}.
In the simpler case that no channel has integral inequality tests,
we obtain just a Petri net, for which reachability is decidable
\cite{Mayr:PN:Reach:STOC:1981,Kosaraju:VASS:STOC:1982,Lambert:TCS:1992,LerouxSchmitz:LICS:2015}
and EXPSPACE-hard \cite{Lipton:1976};
the exact complexity of Petri nets is a long-standing open problem.

\mysubparagraph{Related work.}

\emph{Communicating automata} (\ca) were introduced in the early 80's
as a fundamental model of concurrency \cite{BrandZafiropulo:CFM:ACM:1983,Pachl:CA:1982}.
As a way of circumventing undecidability,
restricting the communication topology to polyforest has been already cited \cite{Pachl:CA:1982,TorreMadhusudanParlato:TACAS:2008}.
Other popular methods include
allowing messages to be nondeterministically lost
\cite{CeceFinkel:IC:1996,AbdullaJonsson:LCS:IC:1996,ChampartSchonebelen:CONCUR:08}
(later generalised to include priorities \cite{HaaseSchmitzSchoebelen:LMCS:2014});
restricting the analysis to half-duplex communication \cite{CeceFinkel:IC:2005}
(later generalised to mutex communication \cite{HeusnerLerouxMuschollSutre:LMCS:2012});
%
%
restricting the communication policy to bounded context switching \cite{TorreMadhusudanParlato:TACAS:2008};
weakening the FIFO semantics to the bag semantics allowing for the reordering of messages \cite{ClementeHerbreteauSutre:CONCUR:2014}.
The model of \ca has been extended in diverse directions,
such as \ca with counters \cite{HeussnerLeGallSutre:2012:FSTTCS},
with stacks \cite{HeusnerLerouxMuschollSutre:LMCS:2012},
lossy \ca with data \cite{AbdullaAiswaryaAtig:LICS:2016},
and time \cite{AbdullaAtigCederberg:TLCS:FSTTCS12}.





\section{Preliminaries}
\label{sec:preliminaries}

Let $\N$ be the set of natural, $\Z$ the integer, $\Q$ the rational, and $\Q_{\geq0}$ the nonnegative rational numbers.
Let $\I := \Q \cap [0, 1)$ be the rational unit interval.
For $a \in \Q$, let $\floor a \in \Z$ and $\fract a \in \I$ denote its integral and, \resp, fractional part;
for $b \in \Q$, let the \emph{cyclic difference} be $a \ominus b = \fract {a - b}$
and the \emph{cyclic addition} be $a \oplus b = \fract {a + b}$.
For $a, k \in \Z$,
let $a \eqv m k$ denote the congruence modulo $m \in \N\setminus\set 0$,
which we extend to $a \in \Q$ by $a \eqv m k$ iff $\floor a \eqv m k$.
For a set of variables $X$ and a domain $A$,
let $A^X$ be the set of valuations for variables in $X$ taking values in $A$.
For a valuation $\mu \in A^X$, a variable $x \in X$, and a new value $a \in A$,
let $\mu[x \mapsto a]$ be the new valuation which assigns $a$ to $x$,
and agrees with $\mu$ on $X \setminus \set x$.
For a subset of variables $Y \subseteq X$, let $\restrict \mu Y \in A^Y$
be the restricted valuation agreeing with $\mu$ on $Y$.
For two disjoint domains $X, Y$ and $\mu \in A^X, \nu \in A^Y$,
let $(\mu \cup \nu) \in A^{X \cup Y}$ be the valuation which agrees with $\mu$ on $X$ and with $\nu$ on $Y$.

\mysubparagraph{Labelled transition systems.}
A \emph{labelled transition system} (\lts) $\A$ is a tuple $\tuple {C, c_I, c_F, A, \to}$
where $C$ is a set of \emph{configurations},
with $c_I, c_F \in C$ two distinguished \emph{initial} and \emph{final} configurations, \resp,
$A$ a set of \emph{actions},
and $\to \,\subseteq\, C \times A \times C$ a \emph{labelled transition relation}.
For simplicity, we write $c \goesto a d$ instead of $(c, a, d)\! \in\, \to$,
and for a sequence of actions $w = a_1 \cdots a_n \in A^*$
we overload this notation as $c \goesto w d$
if there exist intermediate states $c = c_0, c_1, \dots, c_n = d$
s.t., for every $1 \leq i \leq n$, $c_{i-1} \goesto {a_i} c_i$.
For a given LTS $\A$,
the \emph{non-emptiness problem} asks whether there is a sequence of actions $w \in A^*$
\st $c_I \goesto w c_F$.

\mysubparagraph{Clock constraints.}
%
%
%
%
Let $\X$ be a set of clocks of type either \emph{classical} $\x : \Q$, \emph{integral} $\x : \N$, or \emph{fractional} $\x : \I$.
A \emph{clock constraint} over $\X$ is a boolean combination of the \emph{atomic constraints}
\begin{align*}
	&						
		&\textrm{(ineq}&\textrm{uality)}
		&\textrm{(mo}&\textrm{dular)}
		&\textrm{(or}&\textrm{der)} \\
	&\textrm{(non-diagonal)}	
		&\x_0 & \leq k 	\ \  			&\x_0 & \eqv m k  	\ \   	 	&\y_0 & = 0   \\
	&\textrm{(diagonal)}	
		&\x_0 - \x_1 & \leq k	&\x_0 - \x_1 & \eqv m k		&\y_0 & \leq \y_1.
\end{align*}
where $\x_0, \x_1$ are either both classical or integral clocks,
$\y_0, \y_1$ fractional clocks,
$m \in \N$, and $k \in \Z$.
As syntactic sugar we also allow $\true$
and variants with any $\sim \ \in \{\leq, <, \geq, >\}$ in place of $\leq$.
A \emph{clock valuation} 
is a mapping $\mu \in \Qpos^\X$ assigning a non-negative rational number to every clock in $\X$.
%
%
Let $\vec 0$ be the clock valuation $\mu$ \st $\mu(\x) = 0$ for every clock $\x \in \X$.
For a valuation $\mu$ and a clock constraint $\varphi$, 
$\mu$ \emph{satisfies} $\varphi$,
written $\mu \models \varphi$,
if $\varphi$ is satisfied when classical clocks $\x : \Q$ are evaluated as $\mu(\x)$,
integral clocks $\x : \N$ as $\floor{\mu(\x)}$,
and fractional clocks $\y : \I$ as $\fract {\mu(\y)}$.
In particular, $\mu \models (\x_0 - \x_1 \eqv m k)$ is equivalent to
$\floor {\mu(\x_0) - \mu(\x_1)} \eqv m k$
if $\x_0, \x_1 : \Q$ are classical clocks,
and to $\floor {\mu(\x_0)} - \floor{\mu(\x_1)} \eqv m k$
if $\x_0, \x_1 : \N$ are integral clocks.

\ignore{
\begin{remark}[$\floor x - \floor y$ versus $\floor {x - y}$]
	\label{rem:int-diff}
	In the presence of fractional constraints
	the expressive power would not change if,
	instead of atomic constraints $\floor x - \floor y \eqv m k$ and $\floor x - \floor y \leq k$
	speaking of the \emph{difference of the integer parts},
	we would choose $\floor{x - y} \eqv m k$ and  $\floor{x-y} \leq k$
	speaking of the \emph{integer part of the difference},
	since the two are inter-expressible:
	%
	\begin{align}
		\label{eq:floor:fract}
		\floor {x - y} = \floor x - \floor y - \condone {\fract x < \fract y}
			\qquad \textrm{ and } \qquad
		\fract {x - y} &= \fract x - \fract y + \condone {\fract x < \fract y}.
	\end{align}
\end{remark}
}

\mysubparagraph{Timed communicating automata.}
A \emph{communication topology} is a directed graph $\T = \tuple {\P, \C}$
with nodes $\P$ representing \emph{processes}
and edges $\C \subseteq \P \times \P$ representing \emph{channels}
$\pq \in \C$ whenever $\p$ can send messages to $\q$.
%
%
We do not allow multiple channels from $\p$ to $\q$ 
since such a topology would have an undecidable non-emptiness problem (stated below).

A \emph{system of timed communicating automata} (\stca) is a tuple
$\S = \tuple{\T, \M, \family \X \c \C, \family \A \p \P}$
where $\T = \tuple {\P, \C}$ is a communication topology,
$\M$ a finite set of \emph{messages},
$\X^\c$ a set of \emph{channel clocks} for messages sent on channel $\c \in \C$,
and, for every $\p \in \P$,
$\A^\p = \tuple {\L^\p, \ell^\p_I, \ell^\p_F, \X^\p, \Op^\p, \Delta^\p}$
is a \emph{timed communicating automaton} 
with the following components:
$\L^\p$ is a finite set of \emph{control locations},
with $\ell^\p_I, \ell^\p_F \in \L^\p$ two distinguished \emph{initial} and \emph{final} locations therein,
$\X^\p$ a set of \emph{local clocks},
and $\to^\p \subseteq \L^\p \times \Op^\p \times \L^\p$ a set of transitions
of the form $\ell \goesto \op\mbox{}\!^\p \arr$,
where $\op \in \Op^\p$ determines the kind of transition:
\vspace{-1ex}
\begin{itemize}
	\item $\op = \nop$ is a local operation without side effects;
	\item $\op = \elapse$ is a global time elapse operation
	which is executed by all processes at the same time;
	all local and channel clocks evolve at the same rate;
	\item $\op = \testop\varphi$ is a operation testing the values of clocks $\X^\p$
	against the \emph{test constraint} $\varphi$;
	\item $\op = \resetop {\x^\p}$ resets clock $\x^\p \in \X^\p$ to zero;
	\item $\op = \sndop \m \pq \psi$ sends message $\m \in \M$ to process $\q$ over channel $\pq \in \C$;
	the \emph{send constraint} $\psi$ over $\X^\p \cup \X^\pq$ specifies the initial values of channel clocks;
	\item $\op = \rcvop \m \qp \psi$ receives message $\m \in \M$ from process $\q$ via channel $\qp \in \C$;
	the \emph{receive constraint} $\psi$ over $\X^\p \cup \X^\qp$ specifies the final values of channel clocks.
\end{itemize}
We allow transitions
$\p \goesto {\op_1; \dots; \op_n} \q$
containing a sequence of operations as syntactic sugar.
We assume \wlg that test constraints $\varphi$'s are atomic,
that $M$ is the maximal constant used in any inequality or modulo constraint,
that all modular constraints $\eqv M$ are over the same modulus $M$,
that all the sets of local $\X^\P := \family \X \p \P$ and channel clocks $\X^\C := \family \X \c \C$ are disjoint,
and similarly for the sets of locations $\L^\p$ and thus operations $\Op^\p$;
consequently, we can just write $\ell \goesto \op \arr$ without risk of confusion.
A \stca has \emph{untimed channels} if $\X^\C = \emptyset$.
A channel $\c \in \C$ has \emph{inequality tests} if there exists
at least one operation $\sndop \m \c \psi$ or $\rcvop \m \c \psi$
where $\psi$ is an inequality constraint
$\x_0 \sim k$ or $\x_0 - \x_1 \sim k$ over (classical or integral) channel clocks $\x_0, \x_1 \in \X^\C$.
%

	%

\mysubparagraph{Semantics.}
%
%
A \emph{channel valuation} is a family $w = \family w \c \C$
of sequences $w^\c \in (\M \times \Qpos^{\X^\c})^*$ of pairs $(\m, \mu)$,
where $\m$ is a message and $\mu$ is a valuation for channel clocks in $\X^\c$.
%
%
For $\delta \in \Qpos$, let $\mu + \delta$ be the clock valuation which adds $\delta$ to the value of every clock,
i.e., $(\mu + \delta)(\x) := \mu(\x) + \delta$,
and for a channel valuation $w = \family w \c \C$ with $w^\c = (\gamma^\c_1, \mu^\c_1) \cdots (\gamma^\c_{k_\c}, \mu^\c_{k_\c})$
let $w + \delta = \family {w'} \c \C$ where $w'^\c = (\gamma^\c_1, \mu^\c_1 + \delta) \cdots (\gamma^\c_{k_\c}, \mu^\c_{k_\c} + \delta)$.
The semantics of a \stca $\S$ is given as the infinite \lts $\sem \S = \tuple {C, c_I, c_F, A, \goesto{}}$,
where the set of configurations $C$ consists of triples $\tuple {\family \ell \p \P, \mu, \family w \c \C}$
of control locations $\ell^\p$ for every process $\p \in \P$,
a local clock valuation $\mu \in \Qpos^\X$,
and channel valuations $w^\c$ for every channel $\c$;
the initial configuration is $c_I = \tuple {\family {\ell_I} \p \P, \vec 0, (\varepsilon)_{\c \in \C}}$,
where $\ell^\p_I$ is the initial location of $\p$,
all local clocks are initially $0$,
and all channels are initially empty;
similarly, the final configuration is $c_F = \tuple {\family {\ell_F} \p \P, \vec 0, (\varepsilon)_{\c \in \C}}$;
%
the set of actions is $A = \bigcup_{\p \in \P} \Op^\p \cup \Qpos$,
and transitions are determined as follows.
For a duration $\delta \in \Qpos$ we have a transition
\begin{gather}
	\label{eq:elapse}
	\tag{\dag}
 	\tuple {\family \ell \p \P, \mu, u} \goesto \delta \tuple {\family \arr \p \P, \nu, v}
\end{gather}
if \emph{for all} processes $\p$ there is a time elapse transition $\ell^\p \goesto \elapse \arr^\p$,
$\nu = \mu + \delta$,
and $v = u + \delta$.
For an operation $\op \in \Op^\p$, we have a transition
%
${\tuple {\family \ell \p \P, \mu, u = \family u \c \C} \goesto \op \tuple {\family \arr \p \P, \nu, v = \family v \c \C}}$
%
whenever $\p$ has a transition $\ell^\p \goesto \op \arr^\p$,
for every other process $\q \neq \p$ the control location $\arr^\q = \ell^\q$ stays the same,
and $\nu, v$ are determined by a case analysis on $\op$:
\vspace{-2ex}
\begin{itemize}
	\item if $\op = \nop$, then $\nu = \nu$, and $v = u$;
	\item if $\op = \testop\varphi$, then $\mu \models \varphi$, $\nu = \mu$, and $v = u$;
	\item if $\op = \resetop {\x^\p}$, then $\nu = \mu[\x^\p \mapsto 0]$, and $v = u$;
	\item if $\op = \sndop \m \pq \psi$, then $\nu = \mu$,
	there exists a valuation for clock channels $\mu^\pq \in \Qpos^{\X^\pq}$
	\st $ \mu \cup \mu^\pq \models \psi$,
	message $\m$ is added to this channel $v^\pq = (\m, \mu^\pq) \cdot u^\pq$,
	and every other channel $\c \in \C \setminus \set {\pq}$ is unchanged $v^\c = u^\c$;
	\item if $\op = \rcvop \m \qp \psi$, then $\nu = \mu$,
	message $\m$ is removed from this channel $u^\qp = v^\qp \cdot (\m, \mu^\qp)$
	provided that clock channels satisfy $\mu \cup \mu^\qp \models \psi$,
	and every other channel $\c \in \C \setminus \set {\qp}$ is unchanged $v^\c = u^\c$.
\end{itemize}
\stca $\S, \S'$ are \emph{equivalent} if the non-emptiness problem has the same answer for $\sem\S$, $\sem{\S'}$.

\section{Main result}
\label{sec:result}

%
%

We characterise completely which \tca topologies have a decidable non-emptiness problem.
\thmmain

\begin{remark}[Inequality \vs emptiness tests]
	A similar characterisation for \emph{untimed} channels appeared previously in \cite{ClementeHerbreteauStainerSutre:FOSSACS13},
	where channels can be tested for emptiness.
	In that setting, it was shown that non-emptiness of discrete-time \stca with untimed channels
	is decidable precisely for polyforest topologies
	where in each polytree there is at most one channel which can be tested for emptiness.
	Since a timed channel with inequality tests
	can simulate an untimed channel with emptiness tests,
	our decidability result generalises \cite{ClementeHerbreteauStainerSutre:FOSSACS13}
	to the more general case of timed channels,
	and our undecidability result follows from their characterisation.
	The simulation is done as follows.
	Suppose processes $\p, \q$ want to cooperate
	in such a way that $\q$ can test whether the channel $\pq$ is empty.
	Time instants are split between even and odd instants.
	All standard operations of $\p, \q$ are performed at odd instants.
	At even time instants, $\p$ sends to $\q$ a special message $\hat \m$ with initial age $0$
	by performing $\sndop {\hat\m} \pq {\x^\pq = 0}$.
	Process $\q$ simulates an emptiness test on $\pq$
	by receiving message $\hat \m$ with the same age $0$
	$\rcvop {\hat\m} \pq {\x^\pq = 0}$.
	This is indeed correct because if some other message $\m$ was sent by $\p$ afterwards,
	then $\hat \m$ would have age $\geq 1$,
	since all other operations happen at odd instants.
\end{remark}

\begin{proof}[Proof of the ``only if'' direction]
	If the topology is not a polyforest, \ie, it contains an undirected cycle,
	then it is well-known that non-emptiness is undecidable already in the untimed setting
	\cite{BrandZafiropulo:CFM:ACM:1983,Pachl:CA:1982}.
	If the topology is a polyforest, but it contains a polytree with more than one timed channel with integral inequality tests,
	then undecidability follows from \cite[Theorem 3]{ClementeHerbreteauStainerSutre:FOSSACS13}	already in discrete time,
	since non-emptiness tests (on the side of the receiver) can be simulated by timed channels with inequality tests
	as remarked above.
\end{proof}

\mysubparagraph{Plan.}
The rest of the paper is devoted to the decidability proof.
In Sec.~\ref{sec:simple} we simplify the form of constraints.
In Sec.~\ref{sec:desync} we define a {more flexible} \emph{desynchronised} semantics \cite{KrcalYi:CTA:CAV:2006} for the elapse of time,
and in Sec.~\ref{sec:rendezvous} a {more restrictive} \emph{rendezvous} semantics \cite{Pachl:CA:1982} for the exchange of messages.
%
%
Applying these two semantics in tandem allows us to remove channels at the cost of introducing counters (\cf\cite{ClementeHerbreteauStainerSutre:FOSSACS13}).
Notice that fractional constraints are so far kept unchanged.
In Sec.~\ref{sec:register-counter} we introduce \emph{register automata with counters} (\rac)
where registers are used to handle fractional values,
and counters for integer values;
we show that reachability is decidable for \rac.
%
Finally, in Sec.~\ref{sec:translation} we simulate the rendezvous semantics of \stca by \rac.
%
%
Omitted proofs can be found in Sec.~\ref{app}.

\section{Simple \stca}
\label{sec:simple}

A \stca is \emph{simple} if: it contains only integral and fractional clocks;
send constraints are of the form $\x^\c = 0$ (for $\x^\c$ a channel clock);
receive constraints of the form $\x^\c \sim k$, $\x^\c \eqv M k$
for an integral clock $\x^\c : \N$,
and of the form $\y^\pq = \y^\q$
for fractional clocks $\y^\pq, \y^\q : \I$.
We present a non-emptiness preserving transformation of a given \stca into a simple one.

\mysubparagraph{Remove integral clocks.}

We remove integral clocks,
by expressing their constraints as combinations of classical and fractional constraints.
Unlike integral and fractional constraints,
classical constraints $\x - \y \sim k$ with $\x, \y : \Q$
are invariant under the elapse of time.
For every integral clock $\x : \N$,
we introduce a classical $\x_\Q : \Q$ and a fractional clock $\x_\I : \I$
which are reset at the same moment as $\x$.
A constraint $\x - \y \leq k$ on clocks $\x, \y : \N$
is replaced by the equivalent
$(\x_\Q - \y_\Q \leq k \wedge \x_\I \leq \y_\I) \vee \x_\Q - \y_\Q \leq k - 1$.
%
%
%
The same technique can handle modulo constraints and channel clocks.

\ignore{
\mysubparagraph{Uniform channel constraints.}
A send or receive channel constraint $\psi$ is \emph{uniform}
if $\psi$ contains only local and channel clocks of the same type,
\ie, either only classical clocks (for inequality constraints),
or only integral clocks (for modulo constraints),
or only fractional clocks (for order constraints).
We can assume \wlg that there are only uniform (channel) constraints
(test constraints in local transitions are already uniform since they are atomic).
Indeed, let $\op = \sndop m \pq \psi$ be a send operation,
with the channel constraint $\psi$ split as $\psi_\Q \wedge \psimod \wedge \psi_\I$
where $\psi_\Q$ contains only classical,
$\psimod$ only modulo,
and $\psi_\I$ only fractional constraints.
We can split $\op$ as a sequence of three send operations $\sndop m \pq {\psi_\Q}$ (\emph{classical send}),
$\sndop m \pq \psimod$ (\emph{modulo send}), and
$\sndop m \pq {\psi_\I}$ (\emph{fractional send}).
A similar modification is performed on receive operations.
In this way, we can assume \wlg that channel operations are either classical, modulo, or fractional.
}

\mysubparagraph{Copy-send.}

A \stca is \emph{copy-send} if channel clocks are always copies of local clocks of the sender process,
i.e., $\X^\pq = \setof {\hat \x_i^\pq } {\x_i^\p \in \X^\p}$,
and all send constraints of process $\p$ are equal to
\begin{align}
	\copyc \p {} \ \equiv\ \bigwedge_{\x^\p_i \in \X^\p} \hat \x^\pq_i = \x^\p_i.
\end{align}
%

\begin{lemma}
	\label{lemma:copy-send}
	Non-emptiness of \stca's reduces to non-emptiness of copy-send \stca's.
\end{lemma}

\begin{proof}
	Let $\S$ be a \stca. We construct an equivalent copy-send \stca $\S'$
	by letting sender processes $\p$'s send copies of their local clocks
	to receiver processes $\q$'s;
	the latter verifies at the time of reception
	whether there existed suitable initial values for channel clocks of $\S$.
	This transformation relies on the method of \emph{quantifier elimination}
	to show that the guessing of the receiver processes $\q$ can be implemented as constraints.
	
	We perform the following transformation for every channel $\pq \in \C$.
	Let classical local and channel clocks be of the form $\x_i^\p, \x_i^\q, \x_i^\pq : \Q$,
	and let fractional clocks be of the form $\y_i^\p, \y_i^\q, \y_i^\pq : \I$.	
	%
	Consider a pair of transmission (of $\p$) and reception (of $\q$) transitions
	$t^\p = (\ell^\p \goesto {\sndop \m \pq {\psi^\p}} \arr^\p)$ and
	$t^\q = (\ell^\q \goesto {\rcvop \m \qp {\psi^\q}} \arr^\q)$,
	where $\psi^\p, \psi^\q$ are of the form
	\begin{align*}
		\psi^\p \equiv
			&\bigwedge_{(i, j) \in I^\p} \x_i^\pq - \x_j^\p \sim_{ij}^\p k_{ij}^\p
			\,\wedge\!\!\!
			\bigwedge_{(i, j) \in I^\pq} \x_i^\pq - \x_j^\pq \sim_{ij}^\pq k_{ij}^\pq
			\ \wedge & \textrm{ (inequality)}
			\\
			&\bigwedge_{(i, j) \in J^\p} \x_i^\pq - \x_j^\p \eqv M h_{ij}^\p
			\,\wedge\!\!\!
			\bigwedge_{(i, j) \in J^\pq} \x_i^\pq - \x_j^\pq \eqv M h_{ij}^\pq
			\ \wedge & \textrm{ (modular)}
			\\
			&\bigwedge_{(i, j) \in K^\p} \y_i^\pq \approx_{ij}^\p \y_j^\p
			\,\wedge\!\!\!
			\bigwedge_{(i, j) \in K^\pq} \y_i^\pq \approx_{ij}^\pq \y_j^\pq, \textrm { and }
			& \textrm{ (order)}
			\\[1ex]
		\psi^\q \equiv
			&\bigwedge_{(i, j) \in I^\q} \x_i^\pq - \x_j^\q \sim_{ij}^\q k_{ij}^\q
			\,\wedge\!\!\!
			\bigwedge_{(i, j) \in J^\q} \x_i^\pq - \x_j^\q \eqv M h_{ij}^\q
			\,\wedge\!\!\!
			\bigwedge_{(i, j) \in K^\q} \y_i^\pq \approx_{ij}^\q \y_j^\q
			,
	\end{align*}
	with
	${\sim_{ij}^\p, \sim_{ij}^\pq, \sim_{ij}^\q, \approx_{ij}^\p, \approx_{ij}^\pq, \approx_{ij}^\q \,\in\! \set{<, \leq, \geq, >}}$,
	$I^\p, I^\pq, J^\p, J^\pq, K^\p, K^\pq, I^\q, J^\q, K^\q$ sets of pairs of clock indices,
	and $k_{ij}^\p, k_{ij}^\pq, h_{ij}^\p, h_{ij}^\pq, k_{ij}^\q, h_{ij}^\q \in \Z$ integer constants.
	(It suffices to consider diagonal constraints since non-diagonal ones can be simulated.
	We don't consider reception constraints on $\x_i^\pq - \x_j^\pq$
	since they are invariant under time elapse
	and can be checked directly at the time of transmission;
	thence the asymmetry between $\psi^\p$ and $\psi^\q$.)
	In the new copy-send \stca $\S'$, we have a classical channel clock $\hat \x_i^\pq : \Q$
	for every classical local clock $\x_i^\p : \Q$ of $\p$,
	and similarly a new fractional clock $\hat \y_i^\pq : \I$ for every $\y_i^\p : \I$.
	Let $\mu, \nu$ be clock valuations at the time of transmission and reception, respectively.
	The initial value of $\hat \x_i^\pq$ is $\mu(\hat \x_i^\pq) = \mu(\x_i^\p)$.
	We assume the existence of two special clocks $\x_0^\p : \Q, \y_0^\p : \I$ which are always zero upon send,
	i.e., $\mu(\x^\p_0) = \mu(\hat \x^\pq_0) = \mu(\y^\p_0) = \mu(\hat \y^\pq_0) = 0$,
	and thus when the message is received $\nu(\hat \x^\pq_0), \nu(\hat \y^\pq_0)$ equal the total integer, \resp, fractional time that elapsed
	between transmission and reception.
	%
	This allows us to recover, at reception time,
	the initial value of local clocks $\mu(\x_i^\p), \mu(\y_i^\p)$
	and the final value of channel clocks $\nu(\x_i^\pq), \nu(\y_i^\pq)$ as follows: 
	\begin{align}
		\label{eq:mu:nu:1}
		\mu(\x_i^\p) &= \nu(\hat \x_i^\pq) - \nu(\hat \x^\pq_0),
		&\nu(\x_i^\pq) &= \mu(\x_i^\pq) + \nu(\hat \x^\pq_0). 
		\\
		\label{eq:mu:nu:2}
		\mu(\y_i^\p) &= \nu(\hat \y_i^\pq) \ominus \nu(\hat \y^\pq_0),
		&\nu(\y_i^\pq) &= \mu(\y_i^\pq) \oplus \nu(\hat \y^\pq_0).
	\end{align}
	We replace transitions $t^\p, t^\q$ with
	%
		$\ell^\p \goesto {\sndop {\tuple {\m, \psi^\p, \psi^\q}} \pq {\copyc \p {}}} \arr^\p$,
		\resp,
		$\ell^\q \goesto {\rcvop {\tuple {\m, \psi^\p, \psi^\q}} \qp {\psi^\q_0}} \arr^\q$,
		where
	the original message $\m$ is replaced by $\tuple {\m, \psi^\p, \psi^\q}$
	(thus guessing and verifying the correct pair of send-receive constraints $\psi^\p, \psi^\q$),
	the send constraint is the copy constraint $\copyc \p {}$,
	and the new reception formula is $\psi^\q_0 \equiv \exists \bar \x^\pq, \bar \y^\pq \st \psi'^\p \wedge \psi'^\q$
	with $\psi'^\p, \psi'^\q$ obtained from $\psi^\p$, \resp, $\psi^\q$
	by performing the substitutions below (following \eqref{eq:mu:nu:1}, \eqref{eq:mu:nu:2}):
	\begin{align*}
		\x_i^\p \mapsto \hat \x_i^\pq - \hat \x^\pq_0,\qquad
		\x_i^\pq \mapsto \x_i^\pq + \hat \x^\pq_0,\qquad
		\y_i^\p \mapsto \hat \y_i^\pq \ominus \hat \x^\pq_0,\qquad
		\y_i^\pq \mapsto \y_i^\pq \oplus \hat \x^\pq_0.
	\end{align*}
	%
	We can rearrange the conjuncts as
	$\psi^\q_0 \equiv (\exists \bar \x^\pq \st \psi^\q_{\bar \x^\pq}) \wedge (\exists \bar \y^\pq \st \psi^\q_{\bar \y^\pq})$,
	where
	\begin{align*}
	\!\!\!\!\!\!\!\!\psi^\q_{\bar \x^\pq} \ \equiv\ &
				\!\!\!\!\bigwedge_{(i, j) \in I^\p} \!\!\! \x_i^\pq - (\hat \x_j^\pq - \hat \x^\pq_0) \sim_{ij}^\p k_{ij}^\p
				\wedge
				\!\!\!\!\!\bigwedge_{(i, j) \in I^\pq} \!\!\! \x_i^\pq - \x_j^\pq \sim_{ij}^\pq k_{ij}^\pq		
				\wedge
				\!\!\!\!\!\bigwedge_{(i, j) \in I^\q} \!\!\! (\x_i^\pq + \hat \x^\pq_0) - \x_j^\q \sim_{ij}^\q k_{ij}^\q
				\wedge \\
				&
				\!\!\!\!\!\bigwedge_{(i, j) \in J^\p} \!\!\! \x_i^\pq - (\hat \x_j^\pq - \hat \x^\pq_0) \eqv M h_{ij}^\p
				\wedge
				\!\!\!\!\!\bigwedge_{(i, j) \in J^\pq} \!\!\! \x_i^\pq - \x_j^\pq \eqv M h_{ij}^\pq		
				\wedge
				\!\!\!\!\!\bigwedge_{(i, j) \in J^\q} \!\!\! (\x_i^\pq + \hat \x^\pq_0) - \x_j^\q \eqv M h_{ij}^\q
				\\
	\!\!\!\!\!\!\!\!\psi^\q_{\bar \y^\pq} \ \equiv\ &
				\!\!\!\!\!\bigwedge_{(i, j) \in K^\p} \!\!\! \y_i^\pq \approx_{ij}^\p \hat \y_j^\pq \ominus \hat \y^\pq_0
				\wedge
				\!\!\!\!\!\bigwedge_{(i, j) \in K^\pq} \!\!\! \y_i^\pq \approx_{ij}^\pq \y_j^\pq
				\wedge
				\!\!\!\!\!\bigwedge_{(i, j) \in K^\q} \!\!\! \y_i^\pq \oplus \hat \y^\pq_0 \approx_{ij}^\q \y_j^\q.
	\end{align*}
	The formula $\psi^\q_0$ above is not a clock constraint due to the quantifiers.
	Thanks to quantifier elimination,
	we show that it is equivalent to a quantifier-free formula $\widetilde \psi^\q$,
	\ie, a constraint.
	%

	\medskip\noindent
	{\em Classical clocks.}
	We show that $\psi^\q_1 \equiv \exists \bar \x^\pq \st \psi^\q_{\bar \x^\pq}$ is equivalent to a quantifier-free formula
	$\widetilde \psi^\q_{\bar \x^\pq}$.
	%
	%
	By highlighting $\x^\pq_1$,
	we can put $\psi^\q_1$ in the form (we avoid the indices for readability)
	%
	\begin{align*}
	\psi^\q_1 \ \equiv\ \exists \bar\x^\pq \st  
		\psi' \wedge
		\bigwedge u \precsim \x^\pq_1 \wedge \bigwedge \x^\pq_1 \precsim v \wedge
		\bigwedge \x^\pq_1 \eqv M t,
	\end{align*}
	where $\psi'$ does not contain $\x^\pq_1$,
	the $u, v$'s are of one of the three types:
	$(I^\p)$ $k_{1j}^\p + \hat \x_j^\pq - \hat \x^\pq_0$,
	$(I^\pq)$ $k_{1j}^\pq + \x_j^\pq$, or
	$(I^\q)$ $k_{1j}^\q + \x_j^\q - \hat \x^\pq_0$,
	and similarly the $t$'s are of one of the three types
	$(J^\p)$ $h_{1j}^\p + \hat \x_j^\pq - \hat \x^\pq_0$,
	$(J^\pq)$ $h_{1j}^\pq + \x_j^\pq$, or
	$(J^\q)$ $h_{1j}^\q + \x_j^\q - \hat \x^\pq_0$.
	%
	%
	We can now eliminate the existential quantifier on $\x^\pq_1$ and obtain the equivalent formula
	$\psi^\q_2 \equiv \exists \x^\pq_2 \cdots \x^\pq_n \st \psi' \wedge \bigwedge u \precsim v \wedge \bigwedge t \eqv M t'$.
	Atomic formulas $u \precsim v$ in $\psi^\q_2$
	are again of the same types as above:
	If $u : (I^\p), v : (I^\pq)$, then $v - u : (I^\p)$.
	If $u, v : (I^\pq)$, then $v - u : I^\pq$.
	If $u : (I^\q), v : (I^\pq)$, then $v - u : (I^\q)$.
	In any other case, i.e.,
	if $u : (I^\p), (I^\q)$ and $v : (I^\p), (I^\q)$, then
	$u \precsim v$ is already a constraint not containing any $\x^\pq_i$'s
	($\hat x^\pq_0$ appears on both side of each inequality and we can remove it)
	and thus does not participate anymore in the quantifier elimination process.
	The same reasoning applies to the modulo constraints.
	We can thus repeat this process for the other variables $\x^\pq_2, \dots, \x^\pq_n$,
	and we finally get a \emph{constraint} equivalent to $\psi^\q_1$
	of the form $\psi^\q_n \equiv \bigwedge u \precsim v \wedge \bigwedge t \eqv M t'$,
	where the $u, v$'s are of the form $h_{1j}^\p + \hat \x_j^\pq$ or $k_{1j}^\q + \x_j^\q$,
	and similarly the $t, t'$'s are of the form $h_{1j}^\p + \hat \x_j^\pq$ or $h_{1j}^\q + \x_j^\q$.
	Thus, $\psi^\q_n$ is the constraint $\widetilde \psi^\q_{\bar \x^\pq}$ we are after.
	Notice how $\widetilde \psi^\q_{\bar \x^\pq}$ speaks only about new channel clocks $\hat \x_j^\pq$'s
	(which hold copies of $\p$-clocks $\x_j^\p$'s)
	and local $\q$-clocks $\x_j^\q$.
	
	\medskip\noindent
	{\em Fractional clocks.}
	With a similar argument we can show that $\exists \bar \y^\pq \st \psi^\q_{\bar \y^\pq}$
	is equivalent to a quantifier-free formula $\widetilde \psi^\q_{\bar \y^\pq}$;
	the details are presented in App.~\ref{app:simple}.
	To conclude, we have shown that the reception formula $\psi^\q_0$
	is equivalent to the constraint $\widetilde \psi^\q_{\bar \x^\pq} \wedge \widetilde \psi^\q_{\bar \y^\pq}$,
	as required.
\end{proof}

\mysubparagraph{Atomic channel constraints
$\hat \x^\pq = \x^\p,
\ \hat \x^\pq - \x^\q \sim k,
\ \hat \x^\pq - \x^\q \eqv M k,
\ \hat \y^\pq \sim \y^\q$.}
Thanks to the previous part, channel clocks are copies of local clocks.
As a consequence, we can assume \wlg that send and receive constraints are atomic.
Let ${\sndop \m \pq {\copyc \p {}}}$, ${\rcvop \m \qp {\psi^\q_1 \wedge \cdots \wedge \psi^\q_n}}$
be a send-receive pair, where the $\psi^\q_i$'s are atomic.
By sending $n$ times in a row the same message $\m$ as
$\sndop \m \pq {\copyc \p {}}; \dots; \sndop \m \pq {\copyc \p {}}$,
we can split the receive operation into 
$\rcvop \m \qp {\psi^\q_1}; \dots; \rcvop \m \qp {\psi^\q_n}$.
Moreover, if a receive constraints 
uses only $\hat \x^\pq$, or $\hat \y^\pq$ \resp,
then we can assume that the corresponding send constraint is just $\hat \x^\pq = \x^\p$
or, \resp, $\hat \y^\pq = \y^\p$%
---all other channel clocks are irrelevant.
Consequently, all channel constraints can in fact be assumed to be atomic.

\mysubparagraph{Atomic channel constraints
$\hat \x^\pq = 0,\ 
\hat \x^\pq = \hat \x^\q$.}
%

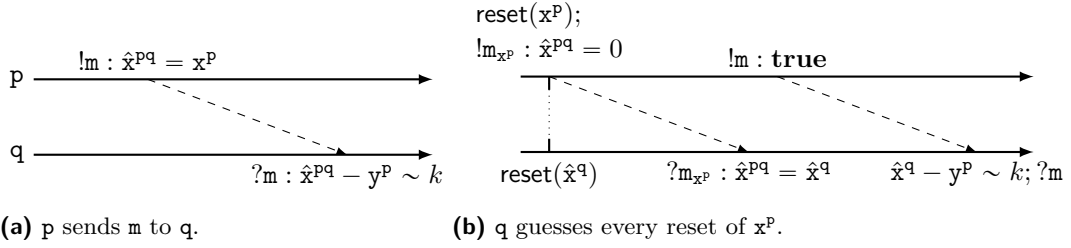
\begin{figure}
	
	\begin{minipage}[b]{.4\linewidth}
		\centering
		\begin{tikzpicture}[x=1.5cm]
			\draw[->,thick,>=latex] (0,0) node [left] {$\p$} -- (3.5,0);
			\draw[->,thick,>=latex] (0,-1) node [left] {$\q$} -- (3.5,-1);
			\draw[->,dashed,>=latex]
				(1.0,0) node [above] {$!\m : \hat \x^\pq = \x^\p$} --
				(2.75,-1) node [below] {$?\m : \hat \x^\pq - \y^\p \sim k$};
		\end{tikzpicture}
		\subcaption{$\p$ sends $\m$ to $\q$.}
		\label{fig:1a}
	\end{minipage}%
	\quad
	\begin{minipage}[b]{.5\linewidth}
		\centering
		\begin{tikzpicture}[x=1.5cm]
			\draw[->,thick,>=latex] (0,0) -- (4.5,0);
			\draw[->,thick,>=latex] (0,-1) -- (4.5,-1);
		
			\draw[thick] (0.25, 0)  node [above] {} -- ++(0, -5pt);
			\draw[thick] (0.25, -1) node [below] {$\resetop{\hat \x^\q}$} -- ++(0, 5pt);
			\draw[dotted] (0.25, 0) -- (0.25, -1);
		
			\draw[->,dashed,>=latex]
				(0.25,0) node [above] {$\begin{array}{ll}\resetop{\x^\p}; \\ !\m_{\x^\p} : \hat \x^\pq = 0 \end{array}$} --
				(2.0,-1) node [below] {$?\m_{\x^\p} : \hat \x^\pq = \hat \x^\q$};
			
			\draw[->,dashed,>=latex]
				(2.25,0) node [above] {$!\m : \true$} --
				(4.0,-1) node [below] {$\hat \x^\q - \y^\p \sim k; ?\m$};
		\end{tikzpicture}
		\subcaption{$\q$ guesses every reset of $\x^\p$.}
		\label{fig:1b}
	\end{minipage}
	\caption{Channel constraints of the form $\hat \x^\pq = 0$ (transmission) and $\hat \x^\pq = \hat \x^\q$ (reception) suffice.}
	\label{fig:zero}
\end{figure}

We further simplify atomic channel constraints
by only sending channel clocks $\hat \x^\pq$ initialised to $0$,
and having receive constraints of the form of equalities $\hat \x^\pq = \hat \x^\q$
between a channel and a local clock;
this holds for both classical and fractional clocks.
%
%
Consider a send/receive pair
{(S)~$\ell^\p \goesto {\sndop \m \pq {\hat \x^\pq = \x^\p}} \arr^\p$}
and
{(R)~$\ell^\q \goesto {\rcvop \m \qp {\psi^\q}} \arr^\q$},
where $\x^\p, \hat \x^\pq$ are either classical or fractional clocks,
and $\psi^\q$ is an atomic constraint 
of the form $\hat \x^\pq - \y^\q \sim k$ or $\hat \x^\pq - \y^\q \eqv M k$ for classical clocks,
or $\hat \x^\pq \sim \y^\q$ for fractional clocks;
cf.~Fig~\ref{fig:zero}.
Process $\p$ communicates to $\q$ every time clock $\x^\p$ is reset
by replacing every reset $\ell^\p_0 \goesto {\resetop {\x_\p}} \arr^\p_0$ with
$\ell^\p_0 \goesto {\resetop {\x^\p};\, \sndop {\m_{\x^\p}} \pq {\hat \x^\pq = 0}} \arr^p_0$
where after the reset $\p$ sends a special message $\m_{\x^\p}$ to $\q$ with initial age $0$.
We add to process $\q$ a copy $\hat \x^\q$ of every clock $\x^\p$ of $\p$;
let $\hat \X^\q$ be the set of these new clocks $\hat \x^\q$'s.
Process $\q$ guesses every reset of $\x^\p$ by resetting its corresponding local clock $\hat \x^\q$
and later verifies that the guess is correct by receiving message $\m_{\x^\p}$ with age equal to $\hat \x^\q$.
Control locations of $\q$ are now of the form $(\ell^\q, \Y)$,
where $\Y \subseteq \hat \X^\q$ is the set of new clocks $\hat \x^\q$'s
for which the reset has been correctly verified.
Initially, $\Y = \hat \X^\q$, \ie,
initially all guesses are correct since all clocks start with value $0$.
For every control location $(\ell^\q, \Y)$ of $\q$,
we have transitions
%
$(\ell^\q, \Y) \goesto {\rcvop {\m_{\x^\p}} \qp {\hat \x^\q = \hat \x^\pq}} (\ell^\q, \Y \cup \set {\hat \x^\q})$
and
%
$(\ell^\q, \Y) \goesto {\resetop {\hat \x^\q}} (\ell^\q, \Y \setminus \set {\hat \x^\q})$.
The original send transition (S) becomes
$\ell^\p \goesto {\sndop \m \pq {\true}} \arr^\p$ with the trivial timing constraint $\true$,
and the original receive transition (R)
becomes an untimed reception $(\ell^\q, \Y) \goesto{\testop{\widetilde\psi^\q}; \rcvop \m \qp {\true}} (\arr^\q, \Y)$
with $\hat \x^\q \in \Y$,
together with a test on local clocks
$\widetilde\psi^\q \equiv \hat \x^\q - \y^\q \sim k$ or, \resp, $\widetilde\psi^\q \equiv \hat \x^\q - \y^\q \eqv M k$ for classical clocks,
or $\widetilde\psi^\q \equiv \hat \x^\q \sim \y^\q$ for fractional clocks.
Constraint $\widetilde\psi^\q$ is now a test on local $\q$-clocks.

\mysubparagraph{Atomic channel constraints $\hat \x^\pq = 0,\ \hat \x^\pq \sim k, \ \hat \x^\pq \eqv M k, \ \hat \y^\pq = \hat \y^\q$.}
By a standard construction,
we can eliminate local diagonal constraints $\x^\q - \y^\q \sim k$ and $\x^\q - \y^\q \eqv M k$
for classical clocks $\x^\q, \y^\q : \Q$
in favour of their non-diagonal counterparts $\x^\q \sim k$ and $\x^\q \eqv M k$ \cite{AlurDill:NTA:TCS:1994}.
%
%
By the previous part, receive channel classical constraints are of the form $\hat \x^\pq = \hat \x^\q$,
and since now the local clock $\hat \x^\q$ participates only in non-diagonal constraints,
what only matters is that $\hat \x^\pq$ and $\hat \x^\q$ are
threshold equivalent for inequality constraints,
and modulo equivalent for modular constraints.
Two clock valuations $\mu, \nu$ are \emph{$M$-threshold equivalent},
written $\mu \approx_\M \nu$ if, for every $\x \in \X^\P$,
$\mu(\x) = \nu(\x)$ if $\mu(\x), \nu(\x) \leq M$,
and $\mu(\x) \geq M$ iff $\nu(\x) \geq M$.
Clearly, if $\mu \approx_\M \nu$, then $\mu \models \varphi$ iff $\nu \models \varphi$
for every constraint $\varphi \equiv \x \sim k$ using constants $k \leq M$.
We can check that $\x, \y$ belong to the same $M$-threshold equivalence class
with the non-diagonal inequality constraint
%
$\varphi_{\approx_\M}(\x, \y) \equiv \bigvee_{k \in \set{0, \dots, M}} \left(\x = k \wedge \y = k \ \vee \ \x \geq M \wedge \y \geq M \right)$.
%
We handle modulo constraints in the same spirit.
Two clock valuations $\mu, \nu$ are \emph{$M$-modulo equivalent},
written $\mu \eqv M \nu$ if, for every $\x \in \X^\P$,
$\mu(\x) \eqv M \nu(\x)$.
Clearly, if $\mu \eqv M \nu$, then $\mu \models \varphi$ iff $\nu \models \varphi$
for every constraint $\varphi \equiv (\x \eqv M k)$.
Moreover, we can check that $\x, \y$ belong to the same $M$-modulo equivalence class
with the non-diagonal modular constraint
%
$\varphi_{\eqv M}(\x, \y) \equiv \bigvee_{k \in \set{0, \dots, M-1}} \left( \x \eqv M k \wedge \y \eqv M k \right)$.
%
Our objective is achieved by replacing classical diagonal reception constraints $\hat \x^\pq = \hat \x^\q$
with the non-diagonal
$\varphi_{\approx_\M}(\hat \x^\pq, \hat \x^\q) \wedge \varphi_{\eqv \M}(\hat \x^\pq, \hat \x^\q)$.
Fractional constraints are untouched in this step.

\mysubparagraph{Remove classical clocks.}

We convert all constraints on classical clocks into equivalent constraints on integral and fractional clocks,
thus undoing the first step of this section.
%
%
For every classical clock $\x : \Q$, we introduce an integral $\x_\N : \N$ and a fractional clock $\x_\I : \I$
which are reset at the same moment as $\x$.
Constraints of the form $\x < k$ are replaced with $\x_\N < k$,
of the form $\x = k$ by $\x_\N = k \wedge \x_\I = 0$,
and of the form $\x > k$ by $\x_\N \geq k + 1 \vee (\x_\N \geq k \wedge \x_\I > 0)$.
It is easy to see that we obtain simple constraints,
as required.

%
%
%


\section{Desynchronised semantics}
\label{sec:desync}

We introduce an alternative run-preserving semantics for \stca, called \emph{desynchronised semantics},
which is the same as the standard semantics
except that time elapse transitions are \emph{local} within processes;
channels $\pq$'s elapse time together with receiving processes $\q$'s.
In order to guarantee that messages are received only after they are sent,
for every channel $\pq$ we allow $\q$ to be ahead of $\p$, but not the other way around.
Thanks to desynchronisation we will remove channels in the next section.
%
%
We make no assumptions on the underlying topology.

Let $\S = \tuple{\T = \tuple {\P, \C}, \M, \family \X \c \C, \family \A \p \P}$ be a \stca.
Assume that for every process $\p \in \P$ there is a special clock $\x^\p_0$
which is never reset and does not appear in any constraint.
The \emph{desynchronised semantics} is the \lts $\desem \S = \tuple {C^\textsf{de}, c_I, c_F^\textsf{de}, A, \desynchto{}}$
where everything is defined as in the standard semantics $\sem S = \tuple {C, c_I, c_F, A, \goesto{}}$,
except $C^\textsf{de}$ which is the restriction of $C$
\begin{align*}
	C^\textsf{de} = \setof
		{ \tuple {\family \ell \p \P, \mu, u} \in C }
		{ \forall \pq \in \C \st \mu(\x^\p_0) \leq \mu(\x^\q_0) }
\end{align*}
ensuring that for every channel $\pq$ process $\q$ is never behind $\p$,
the final configuration is $c_F^\textsf{de} = \tuple {\family {\ell_F} \p \P, \mu_1, (\varepsilon)_{\c \in \C}}$
where $\mu_1(\x^\p) = 0$ for every $\x \in \X \setminus \setof {\x^\p_0} {\p \in \P}$,
%
and for the desynchronised transition relation $\desynchto{}$,
which is defined as $\goesto{}$,
except for the rules for time elapse and transmissions.
For time elapse, \eqref{eq:elapse} is replaced by
%
%
\begin{gather}
	\tag{\ddag}
 	\tuple {\family \ell \p \P, \mu, \family u \c \C} \goesto \delta \tuple {\family \arr \p \P, \nu, \family v \c \C}
\end{gather}
whenever \emph{there exists} a process $\q \in \P$ \st there is a time elapse transition $\ell^\q \goesto \elapse \arr^\q$,
$\restrict \nu {\X^\q} = \restrict \mu {\X^\q} + \delta$, $v^\pq = u^\pq + \delta$ for every channel $\pq \in C$,
for every other process $\p \neq \q$, $\arr^\p = \ell^\p$,
$\restrict \nu {\X^\p} = \restrict \mu {\X^\p}$,
and $v^\c = u^\c$ for every channel $\c$ not of the form $\pq$.
For transmissions, we have the following new rule:
\vspace{-2ex}
\begin{itemize}
	\item $\op = \sndop \m \pq \psi$, $\nu = \mu$,
	there exists a valuation for clock channels $\mu^\pq \in \Qpos^{\X^\pq}$
	\st $(\mu, \mu^\pq) \models \psi$,
	$v^\pq = (\m, \mu^\pq + \delta) \cdot u^\pq$
	where we additionally increase the initial valuation $\mu^\pq$
	by the desynchronisation $\delta := \mu(\x^\q_0) - \mu(\x^\p_0) \geq 0$;
	every other channel $\c \in \C \setminus \set {\pq}$ is unchanged $v^\c = u^\c$.
\end{itemize}
Thanks to the preservation of causality between transmissions and receptions of messages,
the non-emptiness problem for $\sem S$ and $\desem S$ is the same.
\begin{restatable*}[\cf {\cite[Lemma 1]{KrcalYi:CTA:CAV:2006}, \cite[Proposition 1]{ClementeHerbreteauStainerSutre:FOSSACS13}}]{lemma}{lemDesync}
	\label{lem:desync}
	The standard semantics $\sem S$ is equivalent to the desynchronised semantics $\desem S$.
\end{restatable*}

\section{Rendezvous semantics}
\label{sec:rendezvous}

The main advantage of the desynchronised semantics introduced in the previous section is that, over polyforest topologies,
channel operations can be scheduled as too keep the channels always empty.
Moreover, doing this preserves the existence of runs.
This is formalised via the following \emph{rendezvous semantics}:
For a \stca $\S = \tuple{\T = \tuple {\P, \C}, \M, \family \X \c \C, \family \A \p \P}$
define its rendezvous semantics $\rvsem \S = \tuple {C^\textsf{rv}, c_I, c_F, A^\textsf{rv}, \rvto{}}$
to be the restriction of the desynchronised semantics
$\desem \S = \tuple {C, c_I, c_F, A, \desynchto{}}$
where channels are always empty,
%
	$C^\textsf{rv} = \setof
	{ \tuple {\family \ell \p \P, \mu, \family u \c \C} \in C }
	{ \forall \c \in \C \st u^\c = \varepsilon }$,
%
and the transition relation $\rvto{}$ is obtained from $\desynchto{}$
by replacing the two rules for send and receive by the following rendezvous transition
%
%
%
\begin{align*}
	\tuple {\family \ell \p \P, \mu, (\varepsilon)_{\c \in \C}}
		\rvto {(\op^\p, \op^\q)}
			\tuple {\family \arr \p \P, \mu, (\varepsilon)_{\c \in \C}}
\end{align*}
whenever there exists a channel $\pq \in \C$,
a matching pair of send $\ell^\p \goesto {\op^\p} \arr^\p$
and receive transitions $\ell^\q \goesto {\op^\q} \arr^\q$
with $\op^\p = \sndop \m \pq {\psi^\p}, \op^\q = \rcvop \m \qp {\psi^\q}$,
and a valuation for clock channels ${\mu^\pq \in \Qpos^{\X^\pq}}$ \st
${(\mu, \mu^\pq) \models \psi^\p}$ and
${(\mu, \mu^\pq + \delta) \models \psi^\q}$,
where, as in the desynchronised semantics,
$\delta = \mu(\x^\q_0) - \mu(\x^\p_0) \geq 0$
measures the amount of desynchronisation between $\p$ and $\q$;
for every other $\r \in \P \setminus \set{\p, \q}$,
$ \arr^\r = \ell^\r$;
the set of actions $A^\textsf{rv}$ extends $A$ accordingly.

\begin{restatable*}[\cf {\cite{HeusnerLerouxMuschollSutre:LMCS:2012}}]{lemma}{lemRendezvous}
	\label{lem:rendezvous}
	Over polyforest topologies, the desynchronised semantics $\desem S$
	is equivalent to the rendezvous semantics $\rvsem S$.
\end{restatable*}

\section{Register automata with counters}
\label{sec:register-counter}

\begin{wrapfigure}{r}{0.325\textwidth}
	\vspace{-14ex}
	\centering
	\includegraphics[scale=1]{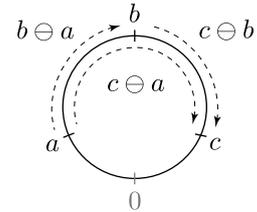}
	\caption{Cyclic order $K(a, b, c)$ \vs cyclic difference $\ominus$. The position of $0$ is irrelevant.}
	\label{fig:K}
	\vspace{-6ex}
\end{wrapfigure}

Thanks to the rendezvous semantics,
we have eliminated the channels,
at the cost of introducing a desynchronisation between senders and receivers.
The integer (unbounded) part of such desynchronisation is modelled by introducing non-negative integer counters;
the fractional part, by registers taking values in $\I = \Q \cap [0, 1)$.

\mysubparagraph{Register constraints.}

Let $\Reg$ be a finite set of \emph{registers}. 
We model fractional values for both local and channel clocks by the \emph{cyclic order} structure ${\K = (\I, K)}$,
where $K \subseteq \I^3$ is the (strict) ternary cyclic order
between rational points $a, b, c \in \I$ in the unit interval,
defined as
%
\begin{align}
	\label{eq:K}
	K(a, b, c) \ &\equiv\ a < b < c \,\vee\, b < c < a \,\vee\, c < a < b.
\end{align}
In other words, $K(a, b, c)$ holds if when moving along a circle of unit length starting from $a$, we first see $b$, and then $c$.
%
%
For $c \in \Q$, we have the relations (\cf~Fig.~\ref{fig:K})
\begin{align}
	\label{eq:ominus}
	b \ominus a \leq c \ominus a \qquad \textrm{ iff } \qquad c \ominus b \leq c \ominus a \qquad \textrm{ iff } \qquad K_0(a, b, c),
\end{align}
where $K_0(a, b, c) \equiv K(a, b, c) \vee a = b \vee b = c$.
%
%
%
%
%
%
A \emph{register constraint} is a quantifier-free formula
%
%
$\varphi$ with variables from $\Reg$ over the vocabulary of $\K$;
since $\K$ admits elimination of quantifiers \cite{survey},
we could allow arbitrary first-order formulas as register constraints without changing the expressiveness of the model.
For a constraint $\varphi$ and a register valuation $r \in \I^\Reg$,
we write $r \models \varphi$ if the formula holds when variables are interpreted according to $r$.

\mysubparagraph{Register automata with counters.}

A \emph{register automaton with counters} (\rac) is a tuple $\R = \tuple {\L, \l_I, \l_F, \Reg, \Cnt, \Delta}$
where $\L$ is a finite set of locations,
$\l_I, \l_F \in \L$ two distinguished initial and final locations therein,
$\Reg$ a finite set of \emph{registers},
$\Cnt$ a finite set of non-negative integer \emph{counters},
and $\Delta$ a finite set of rules of the form $\l \goesto \op \m$
with $\l, \m \in L$,
where $\op$ is either $\nop$,
an \emph{increment} $\incrop \n$ of counter $\n$,
an \emph{decrement} $\decrop \n$ of counter $\n$,
a \emph{counter inequality} $\testop {\n \sim k}$ 
or \emph{modular test} $\testop {\n \eqv m k}$,
a \emph{guess} $\guessop \r$ assigning a new non-deterministic value to register $\r$,
or a \emph{register test} $\testop \varphi$ with $\varphi$ a register constraint.
We allow sequences of operations $\op = (\op_1; \cdots; \op_k)$
and group updates $\incrop {\Cnt'}, \decrop {\Cnt'}$ for $\Cnt' \subseteq \Cnt$ 
as syntactic sugar.
%

\mysubparagraph{Semantics.}

The semantics of a \rac $\R$ as above is the infinite \lts
$\sem \R = \tuple {C, c_I, c_F, A, \to}$
where the set of configurations $C$ consists of tuples $\tuple {\l, n, r}$
with $\l \in \L$ a control location of $\R$,
$n \in \N^\Cnt$ a \emph{counter valuation},
and $r \in \I^\Reg$ a \emph{register valuation},
where the initial configuration is $c_I = \tuple{\l_I, \bar 0, \bar 0}$
with $\bar 0$ the initial counter and (overloaded) register valuation, 
and the final configuration is $c_F = \tuple{\l_F, \bar 0, \bar 0}$.
There is a transition
%
	$\tuple {\l, n, r} \goesto \op \tuple {\m, m, s}$
%
just in case there is a rule $\l \goesto \op \m$ \st
\begin{inparaenum}[(a)]
	\item if $\op = \nop$, then $m = n$ and $s = r$;
	\item if $\op = \incrop \n$, then $m = n[\n \mapsto n(\n)+1]$ and $s = r$;
	\item if $\op = \decrop \n$, then $n(\n) > 0$, $m = n[\n \mapsto n(\n)-1]$, and $s = r$;
	\item if $\op = \testop {\n \sim k}$, then $n(\n) \sim k$, $m = n$, and $s = r$;
	\item if $\op = \testop {\n \eqv m k}$, then $n(\n) \eqv m k$, $m = n$, and $s = r$;
	\item if $\op = \guessop \r$, then $m = n$ and there exists $x \in \I$ \st $s = r[\r \mapsto x]$;
	\item if $\op = \testop \varphi$ with $\varphi$ a register constraint, then $r \models \varphi$, $m = n$, and $s = r$.
\end{inparaenum}
A counter $\n$ appearing in some $\testop {\n \sim k}$ is said to have \emph{inequality tests}.
These can be converted to the well-known zero-tests.
Modular tests $\testop {\n \eqv m k}$ can be removed
by storing in the control location the modulo class of $\n$.
Register tests $\testop \varphi$ can be removed by bookkeeping
a symbolic description of the current register valuation called \emph{orbit}
(similarly as in the region construction for timed automata) \cite{BKL11full}.
\begin{restatable*}{theorem}{thmRacDecidable}
	\label{thm:rac:decidable}
	Non-emptiness is decidable for \rac with $\leq 1$ counter with inequality tests.
\end{restatable*}

\section{Simulating the rendezvous semantics in \rac}
\label{sec:translation}

Let $\S = \tuple{\T = \tuple {\P, \C}, \M, \family \X \c \C, \family \A \p \P}$ be a simple \stca
with $\A^\p = \tuple {L^\p, \ell^\p_I, \ell^\p_F, \X^\p, A^\p, \Delta^\p}$.
We assume that there are neither local diagonal inequality nor modular constraints%
---they can be converted to their non-diagonal counterparts with a standard construction \cite{AlurDill:NTA:TCS:1994}.
%
For every process $\p$,
let $\x^\p_0$ be a \emph{reference clock} which is never reset
representing the ``now'' instant.
We construct a \rac $\R = \tuple {\L, \l_I, \l_F, \Reg, \Cnt, \Delta}$ simulating the rendezvous semantics of $\S$.

\mysubparagraph{From clocks to registers.}

\begin{figure}
	
	\begin{minipage}[b]{.7\linewidth}
		\begin{minipage}[b]{.7\linewidth}
			\begin{tikzpicture}[x=1.5cm]
				\draw[->,thick,>=latex] (0,0) -- (3,0);
				\draw[thick] (0.25, 0) -- ++(0, +5pt) node [above] {$\resetop {\x^\p}$};
				\draw[thick] (1.25, 0) -- ++(0, +5pt) node [above] {$\resetop {\y^\p}$};
				\draw[thick] (2.50, 0) -- ++(0, +5pt) node [above] {$\testop {\x^\p \leq \y^\p}$};
			\end{tikzpicture}
			\subcaption{Clock resets and ordering tests.}
			\label{fig:clocks}
		\end{minipage}
		\\[1ex]
		\begin{minipage}[b]{.7\linewidth}
			\begin{tikzpicture}[x=1.5cm]
				\draw[->,thick,>=latex] (0,0) -- (4.5,0);
				\draw[thick] (0.25, 0) -- ++(0, +5pt) node [above] {$\hat \x^\p := \hat \x_0^\p$} ;
				\draw[thick] (1.50, 0) -- ++(0, +5pt) node [above] {$\hat \y^\p := \hat \x_0^\p$};
				\draw[thick] (3.25, 0) -- ++(0, +5pt) node [above] {$\testop {K_0(\hat \y^\p, \hat \x^\p, \hat \x_0^\p)}$};
			\end{tikzpicture}
			\subcaption{Corresponding register assignments and tests.}
			\label{fig:registers}
		\end{minipage}
	\end{minipage}
	\begin{minipage}[b]{.25\linewidth}
		\centering
		\includegraphics[scale=1]{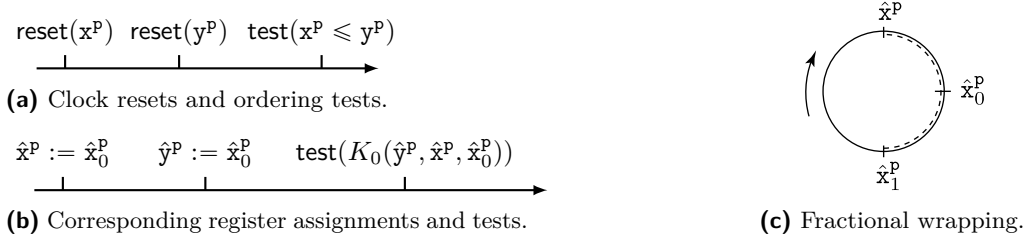}
		\subcaption{Fractional wrapping.}
		\label{fig:wrapped}
	\end{minipage}
	\caption{Fractional clocks $\x^\p, \y^\p : \I$ \vs cyclic registers $\hat \x^\p, \hat \y^\p, \hat \x_0^\p$.}
	\label{fig:clocks:registers}
\end{figure}

%
Fractional clocks $(\x^\p : \I) \in \X^\p$ 
are modelled by a corresponding register $\hat \x^\p \in \Reg$.
A \emph{reference register} $\hat \x^\p_0$ represents the fractional part of the current time of process $\p$;
an auxiliary copy $\hat \x^\p_1$ of the reference register is included to perform the simulation.
The difference between clocks and registers is that a clock $\x^\p$ stores the length of an interval of time%
---the time elapsed since the last reset of $\x^\p$---%
while the corresponding register $\hat\x^\p$ stores the timestamp $\hat \x^\p_0$ when $\x^\p$ was last reset.
In this way, we can express a fractional clock $\x^\p$ as
%
	$\x^\p = \hat \x^\p_0 \ominus \hat \x^\p$.
%
Local and channel fractional constraints
are translated as the following constraints on registers,
for $\x^\p, \y^\p, \x^\pq, \x^\q : \I$:
\begin{align}
	\label{eq:fractional:local}
	&\textrm{[local]} &
	\x^\p \leq \y^\p
	\quad &\textrm{ iff } \quad
	\hat \x^\p_0 \ominus \hat \x^\p \leq \hat \x^\p_0 \ominus \hat \y^\p
	\quad \textrm{ iff } \quad
	K_0(\hat \y^\p, \hat \x^\p, \hat \x^\p_0),
	\\
	\label{eq:fractional:send-receive}
	&\textrm{[send-receive]} &
	\x^\pq \leq \x^\q
	\quad &\textrm{ iff } \quad
	\hat\x^\q_0 \ominus \hat \x^\p_0 \leq \hat\x^\q_0 \ominus \hat\x^\q
	\quad \textrm{ iff } \quad
	K_0(\hat \x^\q, \hat \x^\p_0, \hat \x^\q_0).
\end{align}
%
Intuitively, $\x^\p \leq \y^\p$ holds iff
the last reset of $\y^\p$ happened \emph{before} that of $\x^\p$,
i.e., $K_0(\hat \y^\p, \hat \x^\p, \hat \x^\p_0)$;
\cf Fig.~\ref{fig:clocks}, \ref{fig:registers}.
For \eqref{eq:fractional:send-receive},
when $\p$ sends a message with initial age $0$,
its age at the time of reception is
$\x^\pq = \hat\x^\q_0 \ominus \hat \x^\p_0$,
i.e., the fractional desynchronisation between $\p$ and $\q$.

\mysubparagraph{Unary equivalence.}
We abstract the integral value of clocks into a finite domain called \emph{unary equivalence class}
(akin to the well-known region construction for timed automata).
Let $M \in \N$ be the maximal constant used in any clock constraint of $\S$.
%
%
Two clock valuations $\mu, \nu \in \Qpos^\X$ are \emph{$M$-unary equivalent},
written $\mu \threesim_M \nu$,
if their integral values are threshold $\floor \mu \approx_M \floor \nu$
and modular equivalent $\floor \mu \eqv M \floor \nu$;
\cf Sec.~\ref{sec:simple}.
Let $\Lambda_M$ be the (finite) set of $M$-unary equivalence classes of clock valuations;
for a clock valuation $\mu \in \Qpos^\X$,
let $[\mu] \in \Lambda_M$ be its equivalence class.
%
%
For a set of clocks $\Y \subseteq \X$,
we write $\lambda[\Y \mapsto \Y + 1]$
for the unary class $[\mu']$ of valuations $\mu'$
obtained by taking some valuation $\mu \in \lambda$
and increasing it by $1$ on $\Y$.
If $\mu \threesim_M \nu$ and $\varphi$ contains only inequality and modular constraints on integral clocks
with modulus $M$ and maximal constant $M$,
then $\mu \models \varphi$ iff $\nu \models \varphi$.
We thus overload the notation and for $\lambda \in \Lambda_M$
we write $\lambda \models \varphi$ whenever there exists $\mu \in \lambda$ \st $\mu \models \varphi$.

\mysubparagraph{The translation.}

Control locations in $\L$ are pairs $\l = \tuple {\family \ell \p \P, \lambda}$ of control locations $\ell^\p$ for every $\A^\p$
and a unary equivalence class $\lambda \in \Lambda_M$ abstracting away the values of local integral clocks,
plus additional temporary locations used to perform the simulation.
The initial location is $\l_I = \tuple {\family {\ell_I} \p \P, [\bar 0]}$
and the final location is $\l_F = \tuple {\family {\ell_F} \p \P, [\bar 0]}$.
For each channel $\pq \in \C$
there is a corresponding counter $\n^\pq \in \Cnt$ measuring the amount of integral desynchronisation
between the sender process $\p$ and the receiver process $\q$;
the fractional desynchronisation is measured by $\hat \x^\q_0\ominus\hat \x^\p_0$:
%
%
\begin{align}
	\label{eq:desynch}
	\n^\pq = \floor {\x^\q_0 - \x^\p_0} \qquad \textrm { and } \qquad
	\hat \x^\q_0 \ominus \hat \x^\p_0 = \fract {\x^\q_0 - \x^\p_0}.
\end{align}
Transition rules in $\Delta$ are defined as follows.

	\myitem 1 A transition $\ell^\p \goesto \nop \arr^\p$ 
	is simulated by $\tuple {\family \ell \p \P, \lambda} \goesto \nop \tuple {\family \arr \p \P, \lambda}$ 
	with $\arr^\q = \ell^\q$ $\forall \q \neq \p$.

	\myitem 2 A local time elapse transition $t = \ell^\p \goesto \elapse \arr^\p$ in $\Delta^\p$
	is simulated as follows.
	\vspace{-2ex}
	\begin{enumerate}
		\item[\bf (2a)] We go to a temporary location $\bullet_\lambda$ implicitly depending on $t$:
		%
			$\tuple{\family \ell \p \P, \lambda} \goesto \nop \bullet_\lambda$.
	
		\item[\bf (2b)] We simulate an arbitrary integer time elapse for process $\p$.
		Let $\Cnt^+ = \setof {\n^\qp} {\qp \in \C}$
		be the set of counters corresponding to channels incoming to $\p$
		and let $\Cnt^- = \setof {\n^\pq} {\pq \in \C}$
		for outgoing channels.
		We increase counters $\n^\qp \in \Cnt^+$ by an arbitrary amount,
		and decrease counters $\n^\pq \in \Cnt^-$ by the same amount;
		the unary class $\lambda$ of clocks of $\p$ is updated accordingly:
		%
		For every $\lambda$, we have a transition
		$\bullet_\lambda \goesto {\incrop {\Cnt^+}; \decrop {\Cnt^-}} \bullet_{\lambda'}$,
		where $\lambda' = \lambda[\X^\p \mapsto \X^\p + 1]$.
		These transitions can be repeated an arbitrary number of times.
		
		\item[\bf (2c)] We save the current local time of $\p$ in $\hat \x^\p_1$:
		$\bullet_\lambda \goesto {\guessop {\hat x^\p_1}; \testop {\hat \x^\p_1 = \hat \x^\p_0}} \bullet_{\lambda}^1$.
		
		\item[\bf (2d)] We simulate an arbitrary fractional time elapse for process $\p$
		by guessing a new arbitrary value for the local reference register $\hat \x^\p_0$:
		%
		$\bullet_\lambda^1 \goesto {\guessop {\hat x^\p_0}} \bullet_\lambda^2$.
			
		\item[\bf (2e)] We need to further increase by one
		the integral part of clocks $\x^\p$ whose fractional value was wrapped around $0$ one time more
		than the fractional part of the reference clock $\x^\p_0$.
		For registers $\hat\x, \hat\y, \hat\z$, let
		\begin{align}
			\label{eq:K1-K2}
			K_1(\hat\x, \hat\y, \hat\z) \ \equiv\ K(\hat\x, \hat\y, \hat\z) \vee \hat\y = \hat\z \quad \textrm { and } \quad
			K_2(\hat\x, \hat\y, \hat\z) \ \equiv\ K(\hat\x, \hat\y, \hat\z) \vee \hat\x = \hat\y \neq \z.
		\end{align}
		Cf.~Fig.~\ref{fig:wrapped}:
		Register $\hat \x^\p_1$ stores the old one fractional time.
		%
		%
		In the dashed interval ($\hat \x^\p$ included, $\hat \x^\p_1$ excluded)
		the fractional part of clock $\x^\p$ was wrapped around $0$ one more time than $\x^\p_0$.
		This is the case precisely when $K_1(\hat \x^\p_1, \hat \x^\p, \hat \x^\p_0)$ holds.
		%
		%
		The same adjustment is made for incoming channels $\qp$,
		where $\n^\qp$ must be increased by one	whenever
		$K_1(\hat \x^\p_1, \hat \x^\q_0, \hat \x^\p_0)$	holds.
		For outgoing channels $\pq$, counter $\n^\pq$ must be further decreased by one precisely when
		$K_2(\hat \x^\p_1, \hat \x^\q_0, \hat \x^\p_0)$	holds.
		Let $\SS = \SS^+ \cup \SS^-$, where
		$\SS^+ =	\setof {\hat \x^\p} {\x^\p \in \X^\p} \cup
					\setof {\hat \x^\q_0} {\qp \in \C}$ and
		$\SS^- = 	\setof {\hat \x^\q_0} {\pq \in \C}$,
		be the set of registers that must be checked.
		The automaton guesses a partition $\SS = \SSyes \cup \SSno$
		of those registers corresponding to wrapped clocks and all the others.
		The guess is verified with the formula
		\begin{align}
			\label{eq:cond:phi}
			\varphi\ \equiv\ 
			&\forall \hat\x \in \SSyes \cap \SS^+ \st K_1(\hat \x^\p_1, \hat \x, \hat \x^\p_0) \wedge \forall \hat\x \in \SSyes \cap \SS^- \st K_2(\hat \x^\p_1, \hat \x, \hat \x^\p_0) \wedge \\
			\nonumber
			&\forall \hat\x \in \SSno \cap \SS^+ \st \neg K_1(\hat \x^\p_1, \hat \x, \hat \x^\p_0) \wedge
			\forall \hat\x \in \SSno \cap \SS^- \st \neg K_2(\hat \x^\p_1, \hat \x, \hat \x^\p_0).
		\end{align}
		Let $\Xyes = \setof {\x^\p \in \X^\p} {\hat \x^\p \in \SSyes}$
		be the set of $\p$-clocks whose fractional values were wrapped around $0$.
		The unary class for clocks in $\Xyes$ is updated by $\lambda' = \lambda[\Xyes \mapsto \Xyes + 1]$.
		Let $\Cntyes^+ = \setof {\n^\qp \in \Cnt} {\hat \x^\q_0 \in \SSyes \cap \SS^+}$
		be the set of counters that need to be increased,
		and let $\Cntyes^- = \setof {\n^\pq \in \Cnt} {\hat \x^\q_0 \in \SSyes \cap \SS^-}$
		those that need to be decreased.
		For every $\lambda$ and guessing as above, we have a transition
		$\bullet_\lambda^2 \goesto {\testop{\varphi}; \incrop{(\Cntyes^+)}; \decrop{(\Cntyes^-)}} \bullet_{\lambda'}^3$.
			
		\item[\bf (2f)] The simulation of time elapse terminates with a transition
		$\bullet_{\lambda}^3 \goesto \nop \tuple{\family \arr \p \P, \lambda}$
		for every $\lambda$,
		where $\arr^\q = \ell^\q$ for every other process $\q \neq \p$.	
	\end{enumerate}

	\myitem 3 A test operation $\ell^\p \goesto {\testop \varphi} \arr^\p$ in $\Delta^\p$,
	is simulated by a corresponding transition in $\Delta$
	${\tuple {\family \ell \p \P, \lambda} \goesto \op \tuple {\family \arr \p \P, \lambda}}$.
	An inequality $\varphi \equiv \x^\p \sim k$ or a modular $\varphi \equiv \x^\p \eqv m k$ constraint
	is immediately checked by requiring $\lambda \models \varphi$ and $\op = \nop$.
	Here we use the fact that there are no diagonal inequality or modular constraints in $\S$.
	A fractional constraint $\varphi \equiv \x^\p \leq \y^\p$ on fractional clocks $\x^\p, \y^\p : \I$
	is replaced by the corresponding constraint on fractional registers
	$\op = \testop {K_0(\hat \y^\p, \hat \x^\p, \hat \x^\p_0)}$; \cf~\eqref{eq:fractional:local}.
	For every other $\q \neq \p$, $\arr^\q = \ell^\q$.
	
	\myitem 4 A reset operation $\ell^\p \goesto {\resetop {\x^\p}} \arr^\p$ in $\Delta^\p$
	with $\x^\p : \N$ an integral clock
	is simulated by updating the unary class with the transition
	$\tuple {\family \ell \p \P, \lambda} \goesto \nop \tuple {\family \arr \p \P, \lambda[\x^\p \mapsto [0]]}$ in $\Delta$.
	On the other hand, if $\x^\p : \I$ is a fractional clock,
	then the corresponding register $\hat\x^\p$ records the current timestamp $\hat\x^\p_0$
	by executing
	$\tuple {\family \ell \p \P, \lambda} \goesto {\guessop {\hat\x^\p}; \testop {\hat\x^\p = \hat\x^\p_0}} \tuple {\family \arr \p \P, \lambda[\x^\p \mapsto 0]}$
	in $\Delta$.
	For every other $\q \neq \p$, $\arr^\q = \ell^\q$.

	\myitem 5 A send-receive pair
	$\ell^\p \goesto {\sndop \m \pq {\psi^\p}} \arr^\p$ in $\Delta^\p$
	and $\ell^\q \goesto {\rcvop \m \pq {\psi^\q}} \arr^\q$ in $\Delta^\q$
	is simulated by a test transition
	${\tuple {\family \ell \p \P, \lambda} \goesto {\testop \varphi} \tuple {\family \arr \p \P, \lambda}}$ in $\Delta$,
	where $\arr^\r = \ell^\r$ for every other $\r \in \P \setminus \set {\p, \q}$,
	provided that one of the following conditions hold:
	%
	\begin{enumerate}
	
	\ignore{
		\item If it is a classical send-receive pair,
		then, since our \stca is simple, 
		$\psi^\p, \psi^\q$ are classical (in)equality constraints
		$\psi^\p \equiv \x^\pq = 0$ and $\psi^\q \equiv \x^\pq \sim k$.
		Since the counter $\n^\pq$ measures the integral desynchronisation between $\p$ and $\q$,
		it also measures the integral value of the final age $\x^\pq$ of message $\m$ at the time of reception;
		on the other hand, the fractional value is $\hat\x^\q_0 \ominus \hat\x^\p_0$.
		We distinguish three cases,
		depending on the inequality $\sim$.
		\begin{enumerate}
			\item If $\psi^\q \equiv \x^\pq = k$,
			then we take $\op = \testop {\n^\pq = k \wedge \hat\x^\p_0 = \hat\x^\q_0}$.
			
			\item If $\psi^\q \equiv \x^\pq > k$,
			then we take $\op = \testop {\n^\pq > k \vee (\n^\pq = k \wedge \hat\x^\p_0 \neq \hat\x^\q_0)}$.
			
			\item If $\psi^\q \equiv \x^\pq < k$,
			then we take $\op = \testop {\n^\pq \leq k - 1}$.
		\end{enumerate}
		}
		
		\item[\bf (5a)] If it is an integral send-receive pair,
		then, since our \stca is simple, 
		$\psi^\p, \psi^\q$ are (in)equality constraints of the form
		$\psi^\p \equiv \x^\pq = 0$ and $\psi^\q \equiv \x^\pq \sim k$
		with $\x^\pq : \N$ an integral clock.
		Since the counter $\n^\pq$ measures the integral desynchronisation between $\p$ and $\q$,
		it also measures final value of $\x^\pq$ at the time of reception.
		We take $\varphi \equiv \n^\pq \sim k$.
		
		\item[\bf (5b)] If it is a modular send-receive pair,
		then, since our \stca is simple,
		${\psi^\p \equiv \x^\pq = 0}$ and
		${\psi^\q \equiv (\x^\pq \eqv M k)}$ with $\x^\pq : \N$ an integral clock.
		Take $\varphi \equiv \n^\pq \eqv M k$.
	
		\item[\bf (5c)] The last case is a fractional send-receive pair.
		Since our \stca is simple, we can assume constraints are of the form
		${\psi^\p \equiv \x^\pq = 0}$ and
		${\psi^\q \equiv \x^\pq \leq \x^\q}$ (the other inequality can be treated similarly)
		for fractional clocks $\x^\pq, \x^\q : \I$.
		%
		By~\eqref{eq:fractional:send-receive},
		take $\varphi \equiv K_0(\hat \x^\q, \hat\x^\p_0, \hat\x^\q_0)$.
		%
		%
	\end{enumerate}
%
This concludes the description of the \rac $\R$.

\begin{restatable*}{lemma}{lemTranslation}
	\label{lem:translation}
	The rendezvous semantics $\rvsem \S$ and $\sem \R$ are equivalent.
\end{restatable*}

\begin{proof}[Proof idea]
	We show that the rendezvous semantics of the \stca $\S$
	and the semantics of the \rac $\R$
	are related by a variant of weak bisimulation \cite{Milner:communication:1989}.
	For a configuration $c \in \rvsem \S$ of the form
	$c = \tuple {\family \ell \p \P, \mu}$
	(we ignore the contents of the channels because they are always empty by the definition of rendezvous semantics)
	%
	%
	and a configuration $d \in \sem \R$ of the form
	$d = \tuple {\tuple {\family {\ell'} \p \P, \lambda}, n, r},$
	%
	%
	we say that they are \emph{equivalent},
	written $c \approx d$, if \\[1ex]
	\myitem 1 Control locations are the same: $\ell^\p = \ell'^\p$ for every $\p \in \P$.
	\vspace{-2ex}
		\begin{align}
			\interitem 2 {\vspace{-2ex} The abstraction $\lambda$ is the unary class of the local clock valuation:}
			\label{eq:unary}
			\lambda(\x^\p) &= [\mu(\x^\p)], && \textrm { for every clock } \x^\p \in \X. \\[-1ex]
			\interitem 3 {Register $\hat \x^\p$ keeps track of the fractional part of clock $\x^\p$:}
			\label{eq:registers}
			r(\hat \x^\p_0) \ominus r(\hat \x^\p) &= \fract{\mu(\x^\p)}, && \textrm { for every clock } \x^\p \in \X.
			\interitem 4 {Counter $\n^\pq$ measures the integral desynchronisation between $\p$ and $\q$; \cf\eqref{eq:desynch}:}
			\label{eq:desynch:integral}
			n(\n^\pq) &= \floor {\mu(\x^\q_0) - \mu(\x^\p_0)}, && \textrm { for every channel } \pq \in \C.
			\interitem 5 {
			The fractional desynchronisation between $\p$ and $\q$ is expressed as:}
			\label{eq:desynch:fractional}
			r(\hat \x^\q_0) \ominus r(\hat \x^\p_0) &= \fract {\mu(\x^\q_0) - \mu(\x^\p_0)}, && \textrm { for every channel } \pq \in \C.
		\end{align}
	%
	We show in Sec.~\ref{app:proof} that $c \approx d$ implies that
	the two configurations $c, d$ have the same set of runs starting therein.
	Since the two initial configurations
	are equivalent $c_I \approx d_I$,
	it follows that $\rvsem \S$ is non-empty iff $\sem \R$ is non-empty, as required.
\end{proof}
To sum up, we have so far reduced the non-emptiness problem of a \stca to a simple one (Sec.~\ref{sec:simple}),
then to its rendezvous semantics (Sec.~\ref{sec:desync} and \ref{sec:rendezvous}),
and in this section the latter is reduced to non-emptiness of \rac.
In order to conclude by Theorem~\ref{thm:rac:decidable},
we have to show that, if the communication topology has at most one channel with inequality tests per polytree,
then a \rac with at most one inequality test suffices.
We apply the translation of this section to each polytree (thus obtaining several \rac{s} with at most one inequality test each),
and then simulate the whole polyforest topology by sequentialising each polytree,
which allows to reuse a single inequality test for the entire simulation; \cf Sec.~\ref{app:main} for the details.
To conclude, are able to produce a single \rac with at most one inequality test equivalent to the original \stca.
This finishes the proof of the ``if'' direction of Theorem~\ref{thm:main}.




\bibliographystyle{abbrv}
\bibliography{bib}

\newpage
\appendix


\section{Appendix}
\label{app}

Let $\condone C$, for a condition $C$, be $1$ if $C$ holds, and $0$ otherwise.

\subsection{Missing proof for Sec.~\ref{sec:simple}}
\label{app:simple}

We conclude the proof of Lemma~\ref{lemma:copy-send} in the case of fractional clocks.
\begin{proof}[Second part of the proof of Lemma~\ref{lemma:copy-send}]
	\medskip\noindent
	{\em Fractional clocks.}
	Recall the definition of $\psi^\q_{\bar \y^\pq}$:
	\begin{align*}
		\psi^\q_{\bar \y^\pq} \ \equiv\ &
				\bigwedge_{(i, j) \in K^\p} \y_i^\pq \approx_{ij}^\p \hat \y_j^\pq \ominus \hat \y^\pq_0
				\wedge
				\bigwedge_{(i, j) \in K^\pq} \y_i^\pq \approx_{ij}^\pq \y_j^\pq
				\wedge
				\bigwedge_{(i, j) \in K^\q} \y_i^\pq \oplus \hat \y^\pq_0 \approx_{ij}^\q \y_j^\q.
	\end{align*}
	We show that $\exists \bar \y^\pq \st \psi^\q_{\bar \y^\pq}$ is equivalent to a quantifier-free formula
	$\widetilde \psi^\q_{\bar \y^\pq}$.
	We replace the rightmost atomic formula $\y_i^\pq \oplus \hat \y^\pq_0 \leq \y_j^\q$ in $\psi^\q_{\bar \y^\pq}$
	by an equivalent formula using ``$\ominus$'' instead of ``$\oplus$'';
	the other comparison operators can be dealt with in a similar manner.
	We would like to apply ``$\ominus \hat \y^\pq_0$'' to both sides of the inequality,
	using the obvious fact that $(\y_i^\pq \oplus \hat \y^\pq_0) \ominus \hat \y^\pq_0 = \y_i^\pq$.
	This is safe to do if $\hat \y^\pq_0 \leq \y_i^\pq \oplus \hat \y^\pq_0$ (and thus $\hat \y^\pq_0 \leq \y_j^\q$),
	which is equivalent to $\hat \y^\pq_0 \leq 1 \ominus \y_i^\pq$,
	and we obtain $\y_i^\pq \leq \y_j^\q \ominus \hat \y^\pq_0$ in this case.
	However, if $\y_i^\pq \oplus \hat \y^\pq_0 < \hat \y^\pq_0 \leq \y_j^\q$,
	that is, $1 \ominus \y_i^\pq < \hat \y^\pq_0 \leq \y_j^\q$,
	then the inequality is inverted,
	and we obtain $\y_j^\q \ominus \hat \y^\pq_0 < \y_i^\pq$ in this case.
	Finally, if $\y_j^\q < \hat \y^\pq_0$ (and thus $\y_i^\pq \oplus \hat \y^\pq_0 < \hat \y^\pq_0$),
	then the inequality flips again,
	and we obtain again $\y_i^\pq \leq \y_j^\q \ominus \hat \y^\pq_0$.
	Putting these three cases together,
	we have that $\y_i^\pq \oplus \hat \y^\pq_0 \leq \y_j^\q$ is equivalent to the formula
	\begin{align*}
		&(\y_i^\pq \leq \y_j^\q \ominus \y^\pq_0 \wedge (\hat \y^\pq_0 \leq 1 \ominus \y_i^\pq \vee \y_j^\q < \hat \y^\pq_0)) \vee
		(\y_j^\q \ominus \hat \y^\pq_0 < \y_i^\pq \wedge 1 \ominus \y_i^\pq < \hat \y^\pq_0 \leq \y_j^\q).
	\end{align*}
	We put the $\y_i^\pq$'s in positive positions obtaining the equivalent formula
	\begin{align*}
		&(\y_i^\pq \leq \y_j^\q \ominus \y^\pq_0 \wedge (\y_i^\pq \leq 1 \ominus \hat \y^\pq_0 \vee \y_j^\q < \hat \y^\pq_0)) \vee
		(\y_j^\q \ominus \hat \y^\pq_0 < \y_i^\pq \wedge 1 \ominus \hat \y^\pq_0 < \y_i^\pq \wedge \hat \y^\pq_0 \leq \y_j^\q).
	\end{align*}
	\ignore{
		\begin{gather*}
			(\y_i^\pq \leq \y_j^\q \ominus \hat \y^\pq_0 \wedge
				(\y_i^\pq < 1 \ominus \hat \y^\pq_0 \leq 1 \ominus \y_j^\q \vee
				1 \ominus \y_j^\q < 1 \ominus \hat \y^\pq_0 \leq \y_i^\pq)) \vee \\
			1 \ominus \hat \y^\pq_0 \leq 1 \ominus \y_j^\q \wedge 1 \ominus \hat \y^\pq_0 \leq \y_i^\pq.
		\end{gather*}
	}
	By distributing $\vee$ over $\wedge$,
	we can put $\psi^\q_{\bar \y^\pq}$ in CNF.
	\Wlg it suffices consider a single conjunct thereof,
	which has the general shape (we omit indices for readability)
	\begin{align}
		\label{eq:fract:intermediate}
		\psi \wedge \exists \bar \y^\pq \st \bigwedge u \preceq v,
	\end{align}
	where $\psi$ contains only constraints of the form 
	$\y_j^\q \approx \hat \y^\pq_0$ with $\approx \in \set{<, \leq, \geq, >}$;
	$\preceq \,\in \set {\leq, <}$;
	and the lower $u$'s and upper bound constraints $v$'s are of one the forms
	$\y_i^\pq$, $\hat \y_j^\pq \ominus \hat \y^\pq_0$, $\y_j^\q \ominus \hat \y^\pq_0$, or $1 \ominus \hat \y^\pq_0$.
	By solving \eqref{eq:fract:intermediate} \wrt $\y_1^\pq$,
	we obtain a formula of the form
	\begin{align*}
		\psi \wedge
			\exists	\y_2^\pq \cdots \y_n^\pq \st \varphi \wedge
			\exists \y_1^\pq \st \bigwedge u \preceq \y_1^\pq \wedge \bigwedge \y_1^\pq \preceq v,
	\end{align*}
	where $\psi$ is as in \eqref{eq:fract:intermediate}
	and	$\varphi$ does not contain $\y_1^\pq$.
	%
	%
	%
	By removing the existential quantifier on $\y_1^\pq$ we obtain
	\begin{align*}
		\psi \wedge \exists \y_2^\pq \cdots \y_n^\pq \st \varphi \wedge \bigwedge u \preceq v.
	\end{align*}
	This formula is in the same form as \eqref{eq:fract:intermediate}, but with one quantifier less.
	We can repeat the process and remove all the quantifiers \wrt $\y_2^\pq \dots \y_n^\pq$,
	and obtain a quantifier-free formula of the form $\psi' \wedge \bigwedge u' \preceq v'$
	where $\psi'$ contains only constraints of the form $\y_j^\q \approx \hat \y^\pq_0$ with $\approx \in \set{<, \leq, \geq, >}$,
	and the $u', v'$'s are of one of the forms
	$\hat\y_j^\pq \ominus \hat \y^\pq_0$,
	$\y_j^\q \ominus \hat \y^\pq_0$,
	or $1 \ominus \hat\y^\pq_0$.
	Thus every $u' \preceq v'$ is of the form $a \ominus \hat \y^\pq_0 \preceq b \ominus \hat \y^\pq_0$,
	and by~\eqref{eq:ominus} it can be expressed purely in terms of order constraints on $a, b$.
	We have thus obtained the quantifier-free formula $\widetilde \psi^\q_{\bar \y^\pq}$ we were after.
	Notice that $\widetilde \psi^\q_{\bar \y^\pq}$ speaks only about local $\q$-clocks $\y_j^\q$'s
	and new channel clocks $\hat\y_j^\pq$'s (which hold copies of $\p$-clocks $\y_j^\p$'s).
\end{proof}

\subsection{Missing proofs for Sec.~\ref{sec:desync}}

\lemDesync
\begin{proof}
	Every run in $\sem S$ is also a run in $\desem S$,
	since the latter semantics is a weakening of the former.
	For the other direction, a run in $\desem S$ can be resynchronised by rescheduling all processes $\p$'s
	to execute $\elapse$ transitions at the same time in order for the local now value $\mu_\p(\x^\p_0)$ to be the same for every process. 
	Processes $\p, \q$ in the same polytree are in fact already synchronised $\mu_\p(\x^\p_0) = \mu_\q(\x^\q_0)$
	by the definition of $\desem S$.
	In order to resynchronise a sender process $\p$ with a receiver process $\q$ from another polytree with $\pq \in \C$,
	since $\q$ is always ahead of $\p$ in $\desem S$, 
	in general we need to anticipate the actions of $\p$,
	and in particular transmissions actions.
	This comes at the cost of potentially increasing the length of the contents of channels outgoing from $\p$.
	Since in $\desem S$ the initial value of channel clocks $\mu^\pq$ is automatically advanced
	by the amount of desynchronisation $\delta = \mu(\x^\q_0) - \mu(\x^\p_0) \geq 0$ between sender $\p$ and receiver $\q$,
	in the synchronised run we have $\delta = 0$ and the initial value of channel clocks sent is just $\mu^\pq$.
\end{proof}

\subsection{Missing proofs for Sec.~\ref{sec:rendezvous}}

\lemRendezvous
\begin{proof}
	Every run in $\rvsem S$ is (essentially) a run in $\desem S$ since the former semantics is a strengthening of the latter;
	``essentially'' means that we need to split atomic send/receive operations into a send followed by a receive operation
	in order to properly get a run in $\desem S$.
	For the other direction, it has been shown that on polyforest topologies
	a run in any system of communicating possibly infinite state automata (and in particular in $\desem S$)
	can be rescheduled in order for transmissions to be immediately followed by matching receptions \cite{HeusnerLerouxMuschollSutre:LMCS:2012}.
	By executing these pairs of matching send/receive operations atomically we obtain rendezvous synchronisation.
\end{proof}

\subsection{Missing proofs for Sec.~\ref{sec:register-counter}}

In this section we prove the following theorem.
\thmRacDecidable

First, we introduce some concepts used in the proof.
An \emph{automorphism} of the cyclic order structure $\K = (\I, K)$
is a bijection $\alpha : \I \to \I$ that preserves and reflects $K$,
i.e., $K(a, b, c)$ iff $K(\alpha(a), \alpha(b), \alpha(c))$;
automorphisms are extended point-wise to register valuations $\I^\R$.
The \emph{orbit} of a register valuation $r \in \I^\R$ is the set of valuations $s$
\st there exists an automorphism $\alpha$ transforming $r$ into $s = \alpha(r)$;
the orbit of $r$ is denoted $\orbit r \subseteq \I^\Reg$.
The structure $\K$ is homogeneous~\cite{survey}, 
and thus the set of valuations $\I^\Reg$
is partitioned into exponentially many distinct orbits,
denoted $\orbit {\I^\Reg}$.
We extend the satisfaction relation from valuations $r \models \varphi$
to orbits of valuations $o \in \orbit {\I^\Reg}$,
and write $o \models \varphi$ whenever there exists $r \in o$ \st $r \models \varphi$;
by the definition of orbit, the choice of representative $r$ does not matter.
An orbit, like a region for clock valuations,
is an equivalence class of valuations which are indistinguishable from the point of view of $\K$;
for instance $(0.2, 0.3, 0.7)$, $(0.7, 0.2, 0.3)$, and $(0.8, 0.2, 0.3)$ belong to the same orbit,
while $(0.2, 0.3, 0.3)$ belongs to a different orbit.
We are now ready to prove the theorem above.
\begin{proof}
	Let $\R = \tuple {\L, \l_I, \l_F, \Reg, \Cnt, \Delta}$ be a \rac with maximal constant $M$,
	where we assume \wlg that all modular tests are over the same modulus $M$.
	We construct a \rac without registers $\R'$
	where counters can only be incremented, decremented, and tested for zero
	(i.e., an ordinary counter machine).
	Let $\R' = \tuple {\L', \l_I', \l_F', \Reg', \Cnt', \Delta'}$,
	where the set of locations is $\L' = \L \times \orbit{\I^\Reg} \times \set{0, \dots, M-1}^\Cnt$,
	the initial location is $\l_I' = (\l_I, \orbit {\bar 0}, \bar 0)$,
	the final location is $\l_F' = (\l_F, \orbit {\bar 0}, \bar 0)$,
	the set of registers is empty $\Reg' = \emptyset$,
	the set of counters does not change $\Cnt' = \Cnt$,
	and the set of transition rules $\Delta'$ is defined as follows.
	Let $\l \goesto \op \l'$ be a transition in $\Delta$.
	Then we have one or more transitions in $\Delta'$ of the form
	$(\l, o, \lambda) \goesto {\op'} (\l', o', \lambda')$
	if any of the following conditions is satisfied.
	If $\op = \nop$, then $\op' = \nop, o' = o, \lambda' = \lambda$.
	If $\op = \incrop \n$, then $\op' = \op, o' = o, {\lambda' = \lambda[\n \mapsto (\lambda(\n) + 1) \!\!\mod\! M]}$,
	and if $\op = \decrop \n$, then $\op' = \op, o' = o, {\lambda' = \lambda[\n \mapsto (\lambda(\n) - 1) \!\!\mod\! M]}$.
	If $\op = \testop{\n \leq k}$, then we have the following sequence of transitions for every $0 \leq h \leq k$:
	$\op' = \left((\decrop \n)^h; \testop {\n = 0}; (\incrop \n)^h\right)$, $o' = o, \lambda' = \lambda$.
	Upper bound constraints are thus reduced to ordinary zero tests.
	If $\op = \testop{\n \geq k}$, then we have a sequence of transitions
	$\op' = \left((\decrop \n)^k; (\incrop \n)^k\right)$, $o' = o, \lambda' = \lambda$.
	If $\op = \testop {\n \eqv M k}$, then we have a transition
	$\op' = \nop$, $o' = o, \lambda' = \lambda$, provided that $\lambda \models \n \eqv M k$.
	If $\op = \guessop \r$, then $\op' = \nop$,
	$\lambda' = \lambda$,
	and there is a transition for every orbit $o' \in \orbit {\I^\Reg}$
	which agrees with $o$ on $\Reg \setminus \set \r$,
	and takes an arbitrary value on $\r$,
	i.e., for every $o' \in \orbit {\setof {r'} {r \in o, r'[\r \mapsto r(\r)] = r}}$.
	Finally, if $\op = \testop \varphi$, then there is a transition
	$\op' = \nop$, $o' = o, \lambda' = \lambda$,
	provided that $o \models \varphi$.
	It is standard to show that $\sem \R, \sem {\R'}$ are equivalent \cite{BKL11full}.
	Moreover, if $\R$ has at most one counter with inequality tests,
	then we obtain a counter machine $\R'$ where at most one counter can be tested for zero,
	and the latter model is decidable \cite{Reinhardt:TCS:2008,Bonnet:MFCS:2011}.
\end{proof}

\subsection{Missing proofs for Sec.~\ref{sec:translation}}
\label{app:proof}

\lemTranslation
\begin{proof}
	Assume $c \approx d$. We show two properties of $\approx$.
	Let successor configurations $c', d'$ be of the form
	\begin{align*}
		c' &= \tuple {\family {\ell'} \p \P, \mu'}, \textrm{ and } \\
		d' &= \tuple {\l' = \tuple {\family {\ell'} \p \P, \lambda'}, n', r'}.
	\end{align*}
	\begin{itemize}
		
		\item{} {[\bf Forth property]}
		For every transition $c \rvto \op c'$,
		there is a sequence of transitions $d \to^* d'$
		\st again $c' \approx d'$.
		
		\item{} {[\bf Back property]}
		For every \emph{minimal} sequence of transitions $d \to^* d'$
		there is a transition $c \rvto \op c'$
		\st again $c' \approx d'$.
		Minimality is \wrt the length of any sequence of transitions from $d$
		to any configuration of the form $d'$ above.
		(For instance, we do not allow/take in consideration $d \to^* \tuple {\bullet, n',r'}$
		where the latter is an internal state used during the simulation.)
	\end{itemize}
	It is clear that the initial and final configurations of the two systems are $\approx$-equivalent, and thus by the forth and back properties,
	$\rvsem \S$ and $\sem \R$ are equivalent.
	
	\mysubparagraph{Proof of the forth property.}
	Let $c \rvto \op c'$. We proceed by case analysis on $\op$.
	
	\myitem 1 If $\op = \nop$, then $\mu' = \mu$.
			We take $\lambda' = \lambda$, $n' = n$, $r' = r$.
			Clearly $d \goesto \nop d'$ with $c' \approx d'$.
		
	\myitem 2 Let $\op = \delta \in \Qpos$ be a local time elapse operation for process $\p$.
			%
			%
			Let the amount of time elapsed by $\p$ be $\delta = \mu'(\x^\p_0) - \mu(\x^\p_0) \geq 0$.
			By the definition of desynchronised semantics,
			$\mu'(\x^\p) = \mu(\x^\p) + \delta$ for every clock $\x^\p \in \X^\p$ of $\p$,
			and	$\mu'(\x) = \mu(\x)$ for every other clock $\x \in \X \setminus \X^\p$.
			We show how to update $\lambda, n, r$ accordingly.
			According to the definition of $\R$,
			we start by taking transition 
			$$\tuple{\tuple{\family \ell \p \P, \lambda}, n, r} \goesto \nop \tuple{\bullet_\lambda, n, r}.$$

			We first simulate the integer time elapse $\floor \delta$.
			Recall that $\Cnt^+ = \setof {\n^\qp} {\qp \in \C}$
			is the set of counters corresponding to channels incoming to $\p$
			and $\Cnt^- = \setof {\n^\pq} {\pq \in \C}$ for outgoing channels.
			We increase integer values $\floor \delta$ times, obtaining
			$$\tuple{\bullet_\lambda, n, r} \goesto {(\incrop {\Cnt^+}; \decrop {\Cnt^-})^{\floor \delta}} \tuple{\bullet_{\lambda''}, n'', r},$$
			where $\lambda'' = \lambda[\X^\p \mapsto \X^\p + \floor \delta]$
			and $n'' = n[\Cnt^+ \mapsto \Cnt^+ + \floor \delta, \Cnt^- \mapsto \Cnt^- - \floor \delta]$.
			In order for this transition to be legal,
			it must be the case that for every counter $\n^\pq \in \Cnt^-$,
			$n(\n^\pq) \geq \floor \delta$.
			By the definition of $\delta$ we have
			$\mu(\x^\q_0) - \mu(\x^\p_0) = \mu'(\x^\q_0) - (\mu'(\x^\p_0) - \delta) = \mu'(\x^\q_0) - \mu'(\x^\p_0) + \delta$.
			By the definition of desynchronised semantics,
			$\mu'(\x_0^\q) \geq \mu'(\x_0^\p)$,
			and thus we conclude
			\begin{align}
				\label{eq:desynch:amount}
				\mu(\x^\q_0) - \mu(\x^\p_0) \geq \delta
			\end{align}
			By~\eqref{eq:desynch:integral},
			$n(\n^\pq) = \floor {\mu(\x^\q_0) - \mu(\x^\p_0)}$,
			and thus in particular $n(\n^\pq) \geq \floor \delta$.

			We now simulate the fractional time elapse $\fract \delta$.
			We save the previous value of $\hat\x^\p_0$ in $\hat\x^\p_1$
			and we guess a new fractional ``now'' for process $\p$:
			$$	\tuple{\bullet_{\lambda''}, n'', r} \goesto {
					\guessop {\hat x^\p_1};
					\testop {\hat \x^\p_1 = \hat \x^\p_0};
					\guessop {\hat x^\p_0}}
				\tuple{\bullet_{\lambda''}^2, n'', r'},$$
			where $r' = r[\hat\x^\p_1 \mapsto r(\hat\x^\p_0), \hat\x^\p_0 \mapsto r(\hat\x^\p_0) \oplus \fract\delta]$.
			Eq.~\eqref{eq:registers} is satisfied for $\mu', r'$, since
			\begin{align*}
				\fract{\mu'(\x^\p)} &=
				\fract{\mu(\x^\p) + \delta} =
				\fract{\mu(\x^\p)} \oplus \fract \delta = && \textrm{ by \eqref{eq:registers} } \\
				&= r(\hat \x^\p_0) \ominus r(\hat \x^\p) \oplus \fract \delta =
				r'(\hat \x^\p_0) \ominus r(\hat \x^\p) =
				r'(\hat \x^\p_0) \ominus r'(\hat \x^\p).
			\end{align*}
			Also Eq.~\eqref{eq:desynch:fractional} is satisfied, since
			\begin{align*}
				r'(\hat \x^\q_0) \ominus r'(\hat \x^\p_0) &=
				r(\hat \x^\q_0) \ominus (r(\hat \x^\p_0) \oplus \delta) =
				(r(\hat \x^\q_0) \ominus r(\hat \x^\p_0)) \ominus \delta = && \textrm{ by \eqref{eq:desynch:fractional} } \\
				&= \fract {\mu(\x^\q_0) - \mu(\x^\p_0)} \ominus \delta =
				\fract {\mu'(\x^\q_0) - (\mu'(\x^\p_0) - \delta)} \ominus \delta = \\
				&= \mu'(\x^\q_0) \ominus \mu'(\x^\p_0) \oplus \delta \ominus \delta = \mu'(\x^\q_0) \ominus \mu'(\x^\p_0) = \\
				&= \fract{\mu'(\x^\q_0) - \mu'(\x^\p_0)}.
			\end{align*}

			We now fix the integer value of those clocks whose fractional value was wrapped around zero
			one more time than the fractional value of $\x^\p_0$.
			Let $\X = \X^\p \cup \setof {\x^\pq}{\x^\pq \in \C} \cup \setof {\x^\qp}{\x^\qp \in \C}$
			be the set of all possibly affected clocks.
			The set of local clocks of $\p$ to be further increased by one is $\Xyes$,
			the set of counters for incoming channels to be further increased by one is $\Cntyes^+$,
			and the set of counters for outgoing channels to be further decreased by one is $\Cntyes^-$,
			where:
			\begin{align*}
				\Xyes &= \setof {\x^\p \in \X^\p}
					{\floor {\mu'(\x^\p)} = \floor {\mu(\x^\p)} + \floor \delta + 1}, \\
				\Cntyes^+ &= \setof {\n^\qp \in \Cnt^+}
					{\floor {\mu'(\x^\p_0) - \mu'(\x^\q_0)} = \floor {\mu(\x^\p_0) - \mu(\x^\q_0)} + \floor \delta + 1}, \textrm{ and } \\
				\Cntyes^- &= \setof {\n^\pq \in \Cnt^-}
					{\floor {\mu'(\x^\q_0) - \mu'(\x^\p_0)} = \floor {\mu(\x^\q_0) - \mu(\x^\p_0)} - \floor \delta - 1}.
			\end{align*}
			The set of registers to be checked
			$\SS = \SS^+ \cup \SS^-$ with
			$\SS^+ =	\setof {\hat \x^\p} {\x^\p \in \X^\p} \cup
						\setof {\hat \x^\q_0} {\qp \in \C}$,
			$\SS^- = 	\setof {\hat \x^\q_0} {\pq \in \C}$
			is thus partitioned into $\SS = \SSyes \cup \SSno$,
			where $\SSyes = \setof {\hat\x^\p} {\x^\p \in \Xyes} \cup \setof {\hat\x^\q_0} {\n^\qp \in \Cntyes^+ \textrm{ or } \n^\pq \in \Cntyes^-}$
			and $\SSno = \SS \setminus \SSyes$.
			We take transition
			\begin{align}
				\label{eq:trans:elapse}
				\tuple{\bullet_{\lambda''}^2, n'', r'} \goesto
					{\testop{\varphi}; \incrop{(\Cntyes^+)}; \decrop{(\Cntyes^-)}; \nop}
						\tuple{\tuple{\family \arr \p \P, \lambda'}, n', r'},
			\end{align}
			where $\varphi$ was defined in \eqref{eq:cond:phi},
			$n'=n''[\Cntyes^+ \mapsto \Cntyes^+ + 1, \Cntyes^- \mapsto \Cntyes^- - 1]$,
			and $\lambda' = \lambda''[\Xyes \mapsto \Xyes + 1]$.
			We need to argue that this transition can in fact be taken,
			and that equations \eqref{eq:unary} and \eqref{eq:desynch:integral} hold again for $\lambda', n'$.

			First of all, we argue that $r' \models \varphi$ holds.
			There are three cases to consider.
			\begin{enumerate}

				\item If $\x^\p \in \Xyes \subseteq \SSyes$, then the integral value of $\x^\p$ after time elapse equals
				$\floor {\mu'(\x^\p)} = \floor{\mu(\x^\p) + \delta} = \floor {\mu(\x^\p)} + \floor \delta + 1$,
				which holds precisely when $\fract{\mu(\x^\p)} + \fract\delta \geq 1$.
				By~\eqref{eq:registers} and by the definition of $\delta$,
				$(r(\hat \x^\p_0) \ominus r(\hat \x^\p)) + (r'(\hat \x^\p_0) \ominus r(\hat \x^\p_0)) \geq 1$.
				By the definition of $r'$,
				$(r'(\hat \x^\p_1) \ominus r'(\hat \x^\p)) + (r'(\hat \x^\p_0) \ominus r'(\hat \x^\p_1)) \geq 1$.
				This is equivalent to say that the distance on the unit circle
				of going from $r'(\hat \x^\p)$ to $r'(\hat \x^\p_1)$
				and then from the former to $r'(\hat \x^\p_0)$, is at least one.
				This is the same as saying $K_1(r'(\hat \x^\p_1), r'(\hat \x^\p), r'(\hat \x^\p_0))$
				as defined in \eqref{eq:K1-K2}.

				\item If $\hat \x^\q_0 \in \SSyes \cap \SS^+$ (\ie $\n^\qp \in \Cntyes^+$), then 
				$\floor {\mu'(\x^\p_0) - \mu'(\x^\q_0)} = \floor {(\mu(\x^\p_0) + \delta) - \mu(\x^\q_0)} = \floor {\mu(\x^\p_0) - \mu(\x^\q_0) + \delta} = \floor {\mu(\x^\p_0) - \mu(\x^\q_0)} + \floor \delta + 1$,
				and the last equality holds precisely when $\fract{\mu(\x^\p_0) - \mu(\x^\q_0)} + \fract\delta \geq 1$.
				By~\eqref{eq:desynch:fractional} and by the definition of $\delta$,
				this is equivalent to
				$(r(\hat \x^\p_0) \ominus r(\hat \x^\q_0)) + (r'(\hat \x^\p_0) \ominus r(\hat \x^\p_0)) \geq 1$.
				By the definition of $r'$,
				$(r'(\hat \x^\p_1) \ominus r'(\hat \x^\q_0)) + (r'(\hat \x^\p_0) \ominus r'(\hat \x^\p_1)) \geq 1$,
				which, similarly as before,
				is equivalent to $K_1(r'(\hat \x^\p_1), r'(\hat \x^\q_0), r'(\hat \x^\p_0))$.

				\item The argument for $\hat \x^\q_0 \in \SSyes \cap \SS^-$ (\ie $\n^\pq \in \Cntyes^-$) is analogous:
				$\floor {\mu'(\x^\q_0) - \mu'(\x^\p_0)} = \floor {\mu(\x^\q_0) - (\mu(\x^\p_0) + \delta)} = \floor {\mu(\x^\q_0) - \mu(\x^\p_0) - \delta} = \floor {\mu(\x^\q_0) - \mu(\x^\p_0)} - \floor \delta - 1$.
				Since by the desynchronised semantics $\mu'(\x^\q_0) - \mu'(\x^\p_0) \geq 0$
				and thus $\mu(\x^\q_0) - \mu(\x^\p_0) \geq \delta$,
				the equality $\floor {\mu(\x^\q_0) - \mu(\x^\p_0) - \delta} = \floor {\mu(\x^\q_0) - \mu(\x^\p_0)} - \floor \delta - 1$
				holds precisely when
				$\fract{\mu(\x^\q_0) - \mu(\x^\p_0)} < \fract \delta$.
				By~\eqref{eq:desynch:fractional} and the definition of $\delta$,
				this is equivalent to
				$r(\hat \x^\q_0) \ominus r(\hat \x^\p_0) < r'(\hat \x^\p_0) \ominus r(\hat \x^\p_0)$,
				which by the definition of $r'$ is the same as
				$r'(\hat \x^\q_0) \ominus r'(\hat \x^\p_1) < r'(\hat \x^\p_0) \ominus r'(\hat \x^\p_1)$.
				This is the same as saying that, when going along the unit circle,
				the distance from $r'(\hat \x^\p_1)$ to $r'(\hat \x^\q_0)$
				is strictly smaller than the distance from the same $r'(\hat \x^\p_1)$
				to $r'(\hat \x^\p_0) $, \ie, $K_2(r'(\hat \x^\p_1), r'(\hat \x^\q_0), r'(\hat \x^\p_0))$
				as defined in \eqref{eq:K1-K2}.
			\end{enumerate}

			Since the three arguments in the previous paragraph are equivalences,
			for $\x^\p \in \X^\p \setminus \Xyes$
			$K_1(r'(\hat \x^\p_1), r'(\hat \x^\p), r'(\hat \x^\p_0))$ does not hold.
			Similarly, for $\x^\q_0 \in \SSno \cap \SS^+$,
			$K_1(r'(\hat \x^\p_1), r'(\hat \x^\q_0), r'(\hat \x^\p_0))$ does not hold,
			and for $\hat \x^\q_0 \in \SSno \cap \SS^-$,
			$K_2(r'(\hat \x^\p_1), r'(\hat \x^\q_0), r'(\hat \x^\p_0))$ does not hold.
			This concludes showing that $r' \models \varphi$ holds.

			In order to conclude that the transition \eqref{eq:trans:elapse} can be taken,
			we need to show that counters in $\Cntyes^-$ can be decremented by one,
			\ie, that for every counter $\n^\pq \in \Cntyes^-$,
			$n''(\n^\pq) > 0$.
			By the definition of $n''$,
			this is the same as $n(\n^\pq) > \floor \delta$,
			and by \eqref{eq:desynch:integral} this is equivalent to
			$\floor {\mu(\x^\q_0) - \mu(\x^\p_0)} > \floor \delta$.
			By the definition of $\Cntyes^-$ above,
			$\floor {\mu(\x^\q_0) - \mu(\x^\p_0)} = \floor {\mu'(\x^\q_0) - \mu'(\x^\p_0)} + \floor \delta + 1$,
			and by the definition of desynchronised semantics $\mu'(\x^\q_0) \geq \mu'(\x^\p_0)$,
			and thus $\floor {\mu(\x^\q_0) - \mu(\x^\p_0)} \geq \floor \delta + 1 > \floor \delta$
			as required.

			We finally show that \eqref{eq:unary} and \eqref{eq:desynch:integral} hold again for $\lambda', n'$.
			Consider $\lambda'$
			and we need to show $\lambda'(\x) = [\floor{\mu'(\x)}]$ for every clock $\x \in \X$.
			By the definition of $\lambda'$,
			1) if $\x^\p \in \Xyes$, then $\lambda'(\x^\p) = \lambda(\x^\p) + \floor \delta + 1$,
			2) if $\x^\p \in \X^\p \setminus \Xyes$, then $\lambda'(\x^\p) = \lambda(\x^\p) + \floor \delta$, and
			3) for every other $\x^\q \in \X \setminus \X^\p$, $\lambda'(\x^\q) = \lambda(\x^\q)$.
			By~\eqref{eq:unary} applied to $\lambda$,
			$\lambda(\x) = [\mu(\x)]$.
			Case 3) is immediate since $\mu'(\x^\q) = \mu(\x^\q)$.
			Regarding case 1),
			by definition of $\Xyes$ we have
			$\floor{\mu'(\x^\p)} = \floor{\mu(\x^\p)} + \floor \delta + 1$,
			and by taking the unary class we have
			$[{\mu'(\x^\p)}] = [{\mu(\x^\p)} + \floor \delta + 1] = [{\mu(\x^\p)}] + \floor \delta + 1 = \lambda(\x^\p) + \floor \delta + 1 = \lambda'(\x^\p)$.
			Case 2) is analogous.

			Now consider $n'$, and we need to show $n'(\n^\qr) = \floor{\mu'(\x^\r_0) - \mu'(\x^\q_0)}$
			for every channel $\qr \in \C$.
			For every channel $\qr$ not mentioning $\p \not\in \set {\q, \r}$,
			the claim is immediate since $n'(\n^\qr) = n(\n^\qr)$ by the definition of $n'$,
			and $\mu', \mu$ take the same value on clocks $\x^\r_0$ and, \resp, $\x^\q_0$.
			For every counter $\n^\qp \in \Cntyes^+$ corresponding to an incoming channel $\qp$,
			by definition of $\Cntyes^+$, $n'$, and $n''$ we have
			$\floor {\mu'(\x^\p_0) - \mu'(\x^\q_0)} = \floor {\mu(\x^\p_0) - \mu(\x^\q_0)} + \floor \delta + 1 = n(\n^\qp) + \floor \delta + 1 = n'(\n^\qp)$, as required.
			If $\n^\qp \in \Cnt^+\setminus\Cntyes^+$, then
			$\floor {\mu'(\x^\p_0) - \mu'(\x^\q_0)} = \floor {\mu(\x^\p_0) - \mu(\x^\q_0)} + \floor \delta = n(\n^\qp) + \floor \delta = n'(\n^\qp)$.
			The two cases $\n^\pq \in \Cntyes^-$ and $\n^\pq \in \Cnt^- \setminus \Cntyes^-$ are similar.
		
	\myitem 3 If $\op = \testop \varphi$ is a test transition on local $\p$-clocks,
	then $\mu' = \mu$ and $\mu \models \varphi$.
	Take $\lambda' = \lambda$, $n' = n$, $r' = r$,
	and thus $c' \approx d'$.
	It remains to establish $d \goesto {\op'} d'$.
	There are two cases to consider.
	\begin{enumerate}
		\item In the first case, $\varphi$ is a non-diagonal inequality or modular constraint.
		By the definition of $\approx$,
		the unary class of $\mu$ is $[\mu] = \lambda$,
		and, by the definition of unary equivalence,
		$\lambda \models \varphi$,
		and thus the constraint can be checked by reading the local control state.
		By the definition of $\R$,
		$d \goesto {\op'} d'$ with $\op' = \nop$.
		
		\item In the second case, $\varphi = \x^\p \leq \y^\p$ is a fractional constraint
		with $\x^\p, \y^\p : \I$ fractional clocks.
		By assumption, $\fract{\mu(\x^\p)} \leq \fract{\mu(\y^\p)}$ holds.
		By~\eqref{eq:registers},
		$r(\hat \x^\p_0) \ominus r(\hat \x^\p) \leq r(\hat \x^\p_0) \ominus r(\hat \y^\p)$.
		By the definition of $K_0$,
		it holds that $K_0(r(\hat \y^\p), r(\hat \x^\p), r(\hat \x^\p_0))$;
		\cf\eqref{eq:fractional:local}.
		Thus $d \goesto {\op'} d'$ with $\op' = \testop {K_0(\hat \y^\p, \hat \x^\p, \hat \x^\p_0}$.
	\end{enumerate}	
	
	\myitem 4 If $\op = \resetop {\x^\p}$ is a reset transition,
	then $\mu' = \mu[\x^\p \mapsto 0]$.
	We update the unary class as $\lambda' = \lambda[\x^\p \mapsto 0]$.
	Counters are unchanged $n' = n$.
	There are two cases to consider.
	\begin{enumerate}
		\item If $\x^\p : \N$ is an integral clock,
		then also registers are unchanged $r' = r$
		and we directly have $d \goesto \nop d'$ with $c' \approx d'$.
	
		\item If $\x^\p : \I$ is a fractional clock,
		then we need to update its corresponding register $\hat \x^\p$
		by taking $r' = r[\hat\x^\p \mapsto r(\hat\x^\p_0)]$.
		We execute $\op' = (\guessop {\hat \x^\p}; \testop {\hat \x^\p = \hat \x^\p_0})$
		as per the definition of $\R$,
		and we have $d \goesto {\op'} d'$.
		After the transitions, \eqref{eq:registers} is satisfied
		since $0 = \fract {\mu'(\x^\p)} = r'(\hat \x^\p) \ominus r'(\hat \x^\p_0) = r(\hat \x^\p_0) \ominus r(\hat \x^\p_0) = 0$.
	\end{enumerate}
	
	\myitem 5 If $\op = (\sndop \m \pq {\psi^\p}; \rcvop \m \pq {\psi^\q})$ is a send-receive pair,
	then ${\op^\p = \sndop \m \pq {\psi^\p}}$,
	${\op^\q = \rcvop \m \qp {\psi^\q}}$,
	and clocks are unchanged $\mu' = \mu$.
	Thus, the unary abstraction $\lambda' = \lambda$,
	counters $n' = n$,
	and registers $r' = r$ are also unchanged.
	It is clear that $c' \approx d'$.
	It remains to establish $d \goesto {\op'} d'$ for a suitable choice of $\op'$.
	
	By the definition of rendezvous semantics,
	there exists a valuation for clock channels ${\mu^\pq \in \Qpos^{\X^\pq}}$ \st
	${(\mu, \mu^\pq) \models \psi^\p}$ and
	${(\mu, \mu^\pq + \delta) \models \psi^\q}$
    with $\delta = \mu(\x^\q_0) - \mu(\x^\p_0)$.
	By~\eqref{eq:desynch:integral} and \eqref{eq:desynch:fractional},
	\begin{align}
		\label{eq:delta}
		\floor \delta = n(\n^\pq) \qquad \textrm { and } \qquad
		\fract \delta = r(\hat\x^\q_0) \ominus r(\hat\x^\p_0).
	\end{align}
	There are now three cases to consider.
	\begin{enumerate}
		
		\item[\bf (5a)] In the first case,
		$\psi^\p \equiv \x^\pq = 0$, $\psi^\q \equiv \x^\pq \sim k$,
		and thus $\mu^\pq(\x^\pq) = 0$ and $\floor\delta \sim k$.
		\ignore{
		We distinguish three further cases depending on $\sim$.
		\begin{enumerate}
			\item Let $\psi^\q \equiv \x^\pq = k$.
			Then, $\delta = k$ iff $\floor \delta = k$ (since $k$ is an integer)
			and $\fract \delta = 0$.
			By~\eqref{eq:delta}, $n(\n^\pq) = k$ and $r(\hat\x^\q_0) = r(\hat\x^\p_0)$.
			We take $\op' = \testop {\n^\pq = k \wedge \hat\x^\p_0 = \hat\x^\q_0}$.

			\item Let $\psi^\q \equiv \x^\pq > k$.
			Then, $\delta > k$ iff either $\floor \delta > k$ or $\floor \delta = k$ and $\fract \delta > 0$.
			By~\eqref{eq:delta}, we take $\op' = \testop {\n^\pq > k \vee (\n^\pq = k \wedge \hat\x^\p_0 = \hat\x^\q_0)}$.

			\item If $\psi^\q \equiv \x^\pq < k$,
			then $\delta < k$ iff $\fract \delta \leq k -1$.
			By~\eqref{eq:delta}, we take $\op' = \testop {\n^\pq \leq k - 1}$.				
		\end{enumerate}
		}
		By~\eqref{eq:delta}, $n(\n^\pq) \sim k$.
		We take $\op' = \testop {\n^\pq \sim k}$.
		
		\item[\bf (5b)] In the second case,
		$\psi^\p \equiv \x^\pq = 0$, $\psi^\q \equiv \x^\pq \eqv M k$,
		and thus $\mu^\pq(\x^\pq) = 0$ and $\floor \delta \eqv M k$.
		%
		By~\eqref{eq:delta}, $n(\n^\pq) \eqv M k$.
		we take $\op' = \testop {\n^\pq \eqv M k}$.
		
		\item[\bf (5c)] In the third case,
		$\psi^\p \equiv \x^\pq = 0$, $\psi^\q \equiv \x^\pq \leq \x^\q$
		for fractional clocks $\x^\pq, \x^\q : \I$,
		and thus $\mu^\pq(\x^\pq) = 0$ and $\fract \delta \leq \fract {\mu(\x^\q)}$.
		By~\eqref{eq:registers} and \eqref{eq:delta}, this is the same as
		${r(\hat\x^\q_0) \ominus r(\hat\x^\p_0) \leq r(\hat\x^\q_0) \ominus r(\hat\x^\q)}$
		which by \eqref{eq:ominus} is equivalent to $K_0(r(\hat\x^\q), r(\hat\x^\p_0), r(\hat\x^\q_0))$.
		We take	${\op' = \testop {K_0(\hat \x^\q, \hat\x^\p_0, \hat\x^\q_0)}}$.
	\end{enumerate}
	
	\mysubparagraph{Proof of the back property.}
	Let $d \to^* d'$ be a minimal sequence of transitions.
	By minimality, no intermediate configuration when going from $d$ to $d'$
	is of the form $d' = \tuple{\tuple {\family {\ell'} \p \P, \lambda'}, n', r'}$.
	By inspection of the definition of $\R$,
	we need to consider five distinct cases.
	
	\myitem 1 In the first case, $\R$ is simulating a $\nop$ transition
	$\ell^\p \goesto \nop\!\!\mbox{}^\p \arr^\p$ of $\S$,
	and thus by minimality $d \goesto \nop d'$ in just one step,
	with $\lambda' = \lambda$, $n' = n$, and $r' = r$.
	Consequently, $c \goesto \nop c'$ in $\sem \S$
	with $c' = \tuple {\family {\ell'} \p \P, \mu}$
	where $\ell'^\q = \ell^\q$ for every $\q \in \P \setminus \set \p$,
	and thus $c' \approx d'$ as required.
	
	\myitem 2 In the second case, $\R$ is simulating a local $\elapse$ transition
	$\ell^\p \goesto \elapse\!\!\mbox{}^\p \arr^\p$ of process $\p$.
	This is the most involved case.
	By the definition of $\R$ and by minimality,
	transitions in $d \to^* d'$ decompose as follows:
	\begin{align*}
		d = \tuple{\tuple {\family {\ell} \p \P, \lambda}, n, r} \goesto \nop
		&\tuple {\bullet_\lambda, n, r} \goesto {(\incrop{\Cnt^+}; \decrop{\Cnt^-})^{\floor \delta}}
		\tuple {\bullet_{\lambda''}, n'', r} \goesto {\quad} \\
		\goesto {\guessop {\hat x^\p_1}; \testop {\hat \x^\p_1 = \hat \x^\p_0}}
		&\tuple {\bullet_{\lambda''}^1, n'', r''} \goesto {\guessop {\hat x^\p_0}}
		\tuple {\bullet_{\lambda''}^2, n'', r'} \goesto {\quad} \\
		\goesto {\testop{\varphi}; \incrop{(\Cntyes^+)}; \decrop{(\Cntyes^-)}}
		&\tuple {\bullet_{\lambda'}^3, n', r'} \goesto \nop
		\tuple {\tuple {\family {\ell'} \p \P, \lambda'}, n', r'} = d'
	\end{align*}
	where $\delta \in \Qpos$ it the total elapsed timed that is simulated,
	split into its discrete and fractional part $\delta = \floor \delta + \fract \delta$,
	$\lambda'' = \lambda[\X^\p \mapsto \X^\p + \floor \delta]$,
	$n'' = n[\Cnt^+ \mapsto \Cnt^+ + \floor \delta, \Cnt^- \mapsto \Cnt^- - \floor \delta]$,
	$r'' = r[\hat x^\p_1 \mapsto r(\hat x^\p_0)]$,
	$r' = r''[\hat x^\p_0 \mapsto r(\hat x^\p_0) \oplus \fract \delta]$,
	$r' \models \varphi$,
	$\lambda' = \lambda''[\Xyes \mapsto \Xyes + 1]$,
	and	$n' = n''[\Cntyes^+ \mapsto \Cntyes^+ + 1, \Cntyes^- \mapsto \Cntyes^- - 1]$.
	This is simulated in $\S$ by letting process $\p$ elapse $\delta$ time units
	and thus go to $c' = \tuple {\family {\ell'} \p \P, \mu'}$,
	with $\ell'^\q = \ell^\q$ for every $\q \in \P \setminus \set \p$,		
	where $\mu' = \mu[\forall \x^\p \in \X^\p \st \x^\p \mapsto \mu(\x^\p) + \delta]$
	(including the reference clock $\x^\p_0$).
	We need to show that the time elapse transition above is legal in $\S$,
	which by the desynchronised semantics amounts to establish that
	for every channel $\qr \in \C$, $\mu'(\x^\q_0) \leq \mu'(\x^\r_0)$.
	Since the value of $\x^\p_0$ increased during the time elapse transition,
	$\mu(\x^\q_0) = \mu'(\x^\q_0) \leq \mu'(\x^\p_0) = \mu(\x^\p_0) + \delta$
	is immediately satisfied for incoming channels $\qp \in \C$
	since $\mu(\x^\q_0) \leq \mu(\x^\p_0)$ follows from the fact that $c$ is a legal configuration in $\sem \S$.
	Let $\pq \in \C$ be an outgoing channel and we need to establish $\mu'(\x^\q_0) - \mu'(\x^\p_0) \geq 0$.
	The latter inequality will follow immediately from establishing \eqref{eq:desynch:integral} and \eqref{eq:desynch:fractional}
	for $n', r', \mu'$.
	For fractional parts, we have
	\begin{align*}
		\fract{\mu'(\x^\q_0) - \mu'(\x^\p_0)}
		&= \fract{\mu(\x^\q_0) - (\mu(\x^\p_0) + \delta)}
		= \fract{\mu(\x^\q_0) - \mu(\x^\p_0)} \ominus \delta =
		&&\quad \textrm{(by \eqref{eq:desynch:fractional})} \\
		&= (r(\hat \x^\q_0) \ominus r(\hat \x^\p_0)) \ominus \delta
		= (r'(\hat \x^\q_0) \ominus (r'(\hat \x^\p_0) \ominus \delta)) \ominus \delta = \\
		&= r'(\hat \x^\q_0) \ominus r'(\hat \x^\p_0),
	\end{align*}
	and thus \eqref{eq:desynch:fractional} is again satisfied for $r', \mu'$.
	For integral parts, we consider two cases, depending on whether the channel is incoming or outgoing.
	\begin{enumerate}
		
		\item For an outgoing channel $\pq$,
		\begin{align*}
			&\floor{\mu'(\x^\q_0) - \mu'(\x^\p_0)}
			=\ \floor{\mu(\x^\q_0) - (\mu(\x^\p_0) + \delta)} = \\
			=\ &\floor{\mu(\x^\q_0) - \mu(\x^\p_0) - \delta} = \\
			=\ &\floor{\mu(\x^\q_0) - \mu(\x^\p_0)} - \floor \delta - \condone {\fract{\mu(\x^\q_0) - \mu(\x^\p_0)} < \fract \delta} =
			&&\quad \textrm{(by \eqref{eq:desynch:integral}, \eqref{eq:desynch:fractional})} \\
			=\ &n(\n^\pq) - \floor \delta - \condone {r(\hat \x^\q_0) \ominus r(\hat \x^\p_0) < \fract \delta} =
			&&\quad \textrm{(by the def.~of $\delta$)} \\
			=\ &n(\n^\pq) - \floor \delta - \condone {r(\hat \x^\q_0) \ominus r(\hat \x^\p_0) < r'(\hat \x^\p_0) \ominus r(\hat \x^\p_0)} =
			&&\quad \textrm{(by the def.~of $r'$)} \\
			=\ &n(\n^\pq) - \floor \delta - \condone {r'(\hat \x^\q_0) \ominus r'(\hat \x^\p_1) < r'(\hat \x^\p_0) \ominus r'(\hat \x^\p_1)} =
			&&\quad \textrm{(by the def.~of $K_2$)} \\
			=\ &n(\n^\pq) - \floor \delta - \condone {K_2(r'(\hat\x^\p_1), r'(\hat\x^\q_0), r'(\hat\x^\p_0))} =
			&&\quad \textrm{(by the def.~of $n'$)} \\
			=\ &n'(\n^\pq),
		\end{align*}
		thus showing that \eqref{eq:desynch:integral} is again satisfied for $n', \mu'$ for outgoing channels $\pq$.
	
		\item For an incoming channel $\qp$,
		\begin{align*}
			&\floor{\mu'(\x^\p_0) - \mu'(\x^\q_0)}
			= \floor{\mu(\x^\p_0) + \delta - \mu(\x^\q_0)} = \\
			=\ &\floor{\mu(\x^\p_0) - \mu(\x^\q_0)} + \floor \delta +
				\condone {\fract{\mu(\x^\p_0) - \mu(\x^\q_0)} + \fract \delta \geq 1} =
			&&\quad \textrm{(by \eqref{eq:desynch:integral}, \eqref{eq:desynch:fractional})} \\
			=\ &n(\n^\qp) + \floor \delta +
				\condone {r(\hat \x^\p_0) \ominus r(\hat \x^\q_0)  + \fract \delta \geq 1} =
			&&\quad \textrm{(by the def.~of $\delta$)} \\
			=\ &n(\n^\qp) + \floor \delta +
				\condone {(r(\hat \x^\p_0) \ominus r(\hat \x^\q_0)) + (r'(\hat \x^\p_0) \ominus r(\hat \x^\p_0)) \geq 1} = \\
			=\ &n(\n^\qp) + \floor \delta +
				\condone {r(\hat \x^\p_0) \ominus r(\hat \x^\q_0) \geq r(\hat \x^\p_0) \ominus r'(\hat \x^\p_0)} =
			&&\quad \textrm{(by the def.~of $r'$)} \\
			=\ &n(\n^\qp) + \floor \delta +
				\condone {r'(\hat \x^\p_1) \ominus r'(\hat \x^\q_0) \geq r'(\hat \x^\p_1) \ominus r'(\hat \x^\p_0)} =
			&&\quad \textrm{(by the def.~of $K_1$)} \\
			=\ &n(\n^\qp) + \floor \delta +
				\condone {K_1(r'(\hat \x^\p_1), r'(\hat \x^\q_0), r'(\hat \x^\p_0))} =
			&&\quad \textrm{(by def.~of $n'$)} \\
			=\ &n'(\n^\pq).
		\end{align*}
		thus showing that \eqref{eq:desynch:integral} is again satisfied for $n', \mu'$ for incoming channels $\qp$.
	\end{enumerate}
	
	Also \eqref{eq:registers} holds:
	\begin{align*}
		\fract{\mu'(\x^\p)}
		&= \fract{\mu(\x^\p) + \delta} = \fract{\mu(\x^\p)} \oplus \delta =
		&& \quad \textrm{(by \eqref{eq:registers})} \\
		&= (r(\hat \x^\p_0) \ominus r(\hat \x^\p)) \oplus \delta
		= (r(\hat \x^\p_0) \oplus \delta) \ominus r(\hat \x^\p) = \\
		&= r'(\hat \x^\p_0) \ominus r'(\hat \x^\p).
	\end{align*}
	
	Finally, also \eqref{eq:unary} holds:
	\begin{align*}
		[{\mu'(\x^\p)}]
		=\ &[{\mu(\x^\p) + \delta}] = \\
		=\ &[{\mu(\x^\p)}] + \floor \delta +
			\condone {\fract{\mu(\x^\p)} + \fract \delta \geq 1} =
		&&\quad \textrm{(by \eqref{eq:unary}, \eqref{eq:registers})} \\
		=\ &\lambda(\x^\p) + \floor \delta +
			\condone {r(\hat \x^\p_0) \ominus r(\hat \x^\p) + \fract \delta \geq 1} =
		&&\quad \textrm{(by the def.~of $\delta$)} \\
		=\ &\lambda(\x^\p) + \floor \delta +
			\condone {(r(\hat \x^\p_0) \ominus r(\hat \x^\p)) + (r'(\hat \x^\p_0) \ominus r(\hat \x^\p_0)) \geq 1} =
		&&\quad \textrm{(by the def.~of $r'$)} \\
		=\ &\lambda(\x^\p) + \floor \delta +
				\condone {(r'(\hat \x^\p_1) \ominus r'(\hat \x^\p)) + (r'(\hat \x^\p_0) \ominus r'(\hat \x^\p_1)) \geq 1} = \\
		=\ &\lambda(\x^\p) + \floor \delta +
			\condone {r'(\hat \x^\p_1) \ominus r'(\hat \x^\p) \geq r'(\hat \x^\p_1) \ominus r'(\hat \x^\p_0)} =
		&&\quad \textrm{(by the def.~of $K_1$)} \\
		=\ &\lambda(\x^\p) + \floor \delta +
			\condone {K_1(r'(\hat \x^\p_1), r'(\hat \x^\p), r'(\hat \x^\p_0))} =
		&&\quad \textrm{(by def.~of $\lambda'$)} \\
		=\ &\lambda'(\x^\p).
	\end{align*}
	
	Altogether, this establishes that the transition $c \goesto \delta c'$ is legal
	and that $c' \approx d'$, as required.
	
	\myitem 3 In the third case, $\R$ is simulating a test transition
	$\ell^\p \goesto {\testop\varphi}\!\!\mbox{}^\p \arr^\p$.
	By minimality, $d \goesto \op d'$ 
	with $d' = \tuple{\tuple {\family {\ell'} \p \P, \lambda}, n, r}$
	where $\ell'^\q = \ell^\q$ for every $\q \in \P \setminus \set \p$.
	Accordingly, we take $c' = \tuple {\family {\ell'} \p \P, \mu}$
	and thus $c' \approx d'$ follows immediately from $c \approx d$.
	We need to show that $c \goesto {\testop\varphi} c'$.
	Following the definition of $\R$,
	we proceed by a case analysis on $\op$.
	\begin{enumerate}
		
		\item If $\op = \nop$, then $\varphi \equiv \x^\p \sim k$ or $\varphi \equiv \x^\p \eqv m k$
		and it holds that $\lambda \models \varphi$.
		In the first case, this means that $\lambda(\x^\p) \sim k$ holds.
		By~\eqref{eq:unary},
		$\lambda(\x^\p) = [{\mu(\x^\p)}]$
		and since the unary equivalence is sound when computed \wrt the maximal constant,
		$\floor {\mu(\x^\p)} \sim k$ holds.
		Since $\x^\p : \N$ is an integral clock, $\mu \models \varphi$ holds, as required.
		The reasoning for the second case is analogous.

		\item If $\op = \testop {K_0(\hat \y^\p, \hat \x^\p, \hat \x^\p_0)}$,
		then $\varphi \equiv \x^\p \leq \y^\p$ for fractional clocks $\x^\p, \y^\p : \I$.
		Thus $K_0(r(\hat \y^\p), r(\hat \x^\p), r(\hat \x^\p_0))$ holds.
		By~\eqref{eq:fractional:local},
		$r(\hat \x^\p_0) \ominus r(\hat \x^\p) \leq r(\hat \x^\p_0) \ominus r(\hat \y^\p)$.
		By~\eqref{eq:registers},
		$\fract{\mu(\x^\p)} \leq \fract{\mu(\y^\p)}$,
		and thus $\mu \models \varphi$ holds, as required.
	\end{enumerate}
	

	\myitem 4 In the fourth case, $\R$ is simulating a reset transition
	$\ell^\p \goesto {\resetop{\x^\p}}\!\!\mbox{}^\p \arr^\p$
	for a clock $\x^\p$ of process $\p$.
	By minimality, $d \goesto \op d'$ 
	with $d' = \tuple{\tuple {\family {\ell'} \p \P, \lambda'}, n, r'}$
	where $\ell'^\q = \ell^\q$ for every $\q \in \P \setminus \set \p$.
	We take $c' = \tuple {\family {\ell'} \p \P, \mu'}$
	with $\mu' = \mu[\x^\p \mapsto 0]$.
	Clearly, $c \goesto {\resetop{\x^\p}} c'$ holds.
	In order to show that \eqref{eq:unary}, \eqref{eq:registers} hold again for $\lambda', r', \mu'$,
	we do a case analysis on $\op$.
	\begin{enumerate}
		
		\item In the first case, $\op = \nop$.
		By the definition of $\R$,
		$\x^\p : \N$ is an integral clock
		and $\lambda' = \lambda[\x^\p \mapsto 0]$ and $r' = r$.
		Obviously $\lambda'(\x^\p) = [0] = [{\mu'(\x^\p)}]$,
		and for every other clock $\x^\q \neq \x^\p$,
		$\lambda'(\x^\p) = \lambda(\x^\p) = \textrm{(by~\eqref{eq:unary})} = [{\mu(\x^\q)}] = [{\mu'(\x^\q)}]$.
		Thus, \eqref{eq:unary} holds again for $\lambda', \mu'$.
		(That \eqref{eq:registers} holds is trivial since $r' = r$ and $\mu'(\x^\q) = \mu(\x^\q)$ for fractional clocks $\x^\q : \I$.)
		
		\item In the second case, $\op = (\guessop {\hat\x^\p}; \testop {\hat\x^\p = \hat\x^\p_0})$.
		Consequently, $r' = r[\hat\x^\p \mapsto r(\hat\x^\p_0)]$.
		Therefore, $r'(\hat \x^\p_0) \ominus r'(\hat \x^\p) = r(\hat \x^\p_0) \ominus r(\hat \x^\p_0) = 0 = \fract{\mu'(\x^\p)}$.
		Thus, \eqref{eq:registers} holds again for $r', \mu'$.
		(That \eqref{eq:unary} holds is trivial since $\lambda' = \lambda$ and $\mu'(\x^\q) = \mu(\x^\q)$ for integral clocks $\x^\q : \N$.)
	\end{enumerate}
	
	\myitem 5 In the fifth, and last case, $\R$ simulates a send-receive pair of transitions
	${\ell^\p \goesto {\op^\p}\!\!\mbox{}^\p \arr^\p}$ of $\p$ with $\op^\p = \sndop m \pq {\psi^\p}$ and
	${\ell^\q \goesto {\op^\q}\!\!\mbox{}^\q \arr^\q}$ of $\q$ with $\op^\q = \rcvop m \pq {\psi^\q}$.
	By the definition of $\R$ and by minimality,
	$d \goesto {\testop(\varphi)} d'$
	with $d' = \tuple{\tuple {\family {\ell'} \p \P, \lambda}, n, r}$
	where $\ell'^\r = \ell^\r$ for every $\r \in \P \setminus \set {\p, \q}$.
	We take $c' = \tuple {\family {\ell'} \p \P, \mu}$
	and we need to argue that in $\rvsem \S$
	we can take the rendezvous transition $c \goesto {(\op^\p, \op^\q)} c'$.
	Let $\delta = \mu(\x^\q_0) - \mu^(\x^\p_0) \geq 0$ be the desynchronisation between sender and receiver.
	Following the definition of desynchronised semantics,
	we need to show that there exists a valuation for clock channels ${\mu^\pq \in \Qpos^{\X^\pq}}$ \st
	${(\mu, \mu^\pq) \models \psi^\p}$ and
	${(\mu, \mu^\pq + \delta) \models \psi^\q}$.
	We proceed by a case analysis on the condition $\varphi$.
	\begin{enumerate}
		
		\item[\bf (5a)] In the first case, $\varphi \equiv \n^\pq \sim k$ is an inequality counter constraint,
		and thus $n(\n^\pq) \sim k$ holds.
		Then, $\psi^\p \equiv \x^\pq = 0$ and $\psi^\q \equiv \x^\pq \sim k$ with $\x^\pq : \N$ an integral clock.
		Take $\mu^\pq(\x^\pq) = 0$. Clearly ${(\mu, \mu^\pq) \models \psi^\p}$ is satisfied.
		By~\eqref{eq:desynch:integral}, $n(\n^\pq) = \floor {\mu(\x^\q_0) - \mu(\x^\p_0)} = \floor \delta$
		and thus $\floor \delta = \floor {\mu^\pq(\x^\pq) + \delta} \sim k$ holds.
		Since $\x^\pq : \N$ is an integral clock,
		the latter is equivalent to $\mu^\pq(\x^\pq) + \delta \models \x^\pq \sim k$,
		thus showing ${(\mu, \mu^\pq + \delta) \models \x^\pq \sim k}$, as required.
		
		\item[\bf (5b)] In the second case, $\varphi \equiv \n^\pq \eqv M k$ is a modular counter constraint,
		and we reason as above.
		
		\item[\bf (5c)] In the last case, $\varphi \equiv K_0(\hat \x^\q, \hat\x^\p_0, \hat\x^\q_0)$ is a register constraint;
		thus $K_0(r(\hat \x^\q), r(\hat\x^\p_0), r(\hat\x^\q_0))$ holds.
		Then, ${\psi^\p \equiv \x^\pq = 0}$ and	${\psi^\q \equiv \x^\pq \leq \x^\q}$
		with $\x^\pq, \x^\q : \I$ two fractional clocks.
		Take $\mu^\pq(\x^\pq) = 0$. Clearly ${(\mu, \mu^\pq) \models \psi^\p}$ is satisfied.
		By the definition of $K_0$,
		$r(\hat \x^\q_0) \ominus r(\hat \x^\p_0) \leq r(\hat \x^\q_0) \ominus r(\hat \x^\q)$ (\cf\eqref{eq:fractional:send-receive}).
		By~\eqref{eq:desynch:fractional}, $r(\hat \x^\q_0) \ominus r(\hat \x^\p_0) = \fract \delta$,
		and by~\eqref{eq:registers}, $r(\hat \x^\q_0) \ominus r(\hat \x^\q) = \fract{\mu(\x^\q)}$.
		Thus, $\fract \delta \leq \fract{\mu(\x^\q)}$,
		that is $(\mu, \mu^\pq + \delta) \models \x^\pq \leq \x^\q$,
		as required.\qedhere
		
	\end{enumerate}
		
\ignore{	

\begin{equation}
	\label{eq:registers}
	r(\hat \x^\p_0) \ominus r(\hat \x^\p) = \fract{\mu(\x^\p)}.
\end{equation}
\item The fractional part of reference clocks $\x^\p_0$'s is stored in the reference registers $\hat \x^\p_0$,
${\fract{\mu(\x^\p_0)} = r(\hat \x^\p_0)}$,
and consequently we can recover also the fractional desynchronisation between $\p$ and $\q$
according to \eqref{eq:desynch} as:
\begin{equation}
	\label{eq:desynch:fractional}
	r(\hat \x^\q_0) \ominus r(\hat \x^\p_0) = \fract {\mu(\x^\q_0) - \mu(\x^\p_0)}.
\end{equation}

			\begin{enumerate}

				\item The last case is a fractional send-receive pair.
				Since our \stca is simple,
				By~\eqref{eq:fractional:send-receive},
				we take .
			\end{enumerate}
}
\end{proof}

\subsection{Missing proofs for Sec.~\ref{sec:result}}
\label{app:main}

\begin{proof}[Proof of the ``if'' direction of Theorem~\ref{thm:main}]
	Let $\S$ be a \stca over a polyforest topology,
	where in each polytree there is at most one channel with integral inequality tests.
	By Lemma~\ref{lem:desync} the standard semantics $\sem \S$
	is equivalent to the desynchronised one $\desem \S$,
	which in turn is equivalent to the rendezvous one $\rvsem \S$
	by Lemma~\ref{lem:rendezvous}.
	By the transformations of Sec.~\ref{sec:simple}
	we an assume that the \stca is simple.
	This allows us to apply the construction of this section
	in order to build a \rac $\sem \R$
	\st the rendezvous semantics $\rvsem \S$ is equivalent to $\sem \R$ by Lemma~\ref{lem:translation}.
	Suppose the topology $\T$ decomposes into $n$ disjoint polytrees $\T_1, \dots, \T_n$,
	where by assumption in each of the $\T_i$'s there is at most one channel with integral inequality tests.
	We obtain a \rac $\R$ with counters of which $n$ have threshold tests,
	and thus unless $n = 1$ we cannot apply immediately Theorem~\ref{thm:rac:decidable} to obtain decidability of the non-emptiness problem.
	With a small modification of the construction of $\R$
	instead of simulating all polytrees $\T_1, \dots, \T_n$ in parallel,
	we can simulate them sequentially by running $\T_1$ first,
	followed by $\T_2$, \dots, till $\T_n$ (\cf \cite[Theorem 3]{ClementeHerbreteauStainerSutre:FOSSACS13}).
	In order for the sequential simulation to be faithful,
	we need to ensure that the same total amount of time elapses when simulating any of the $\T_i$'s.
	For the integral part of the elapsed time, we can add an extra counter $\n_{\T_i}$ for each component
	which is increased by one every time some fixed process $\p$ therein elapses 1 time unit;
	at the end of all simulations, we additionally check that $\n_{\T_1} = \cdots = \n_{\T_n}$
	by decreasing all such counters by $1$ until they all hit $0$.
	(Notice that at the end of the simulation of $\T_i$ all processes therein elapse the same amount of time
	since we require all counters $\n^\pq$ to be $0$ at the end of the run.)
	For the fractional part of the elapsed time no additional check is needed,
	since reference registers $\hat \x^\p_0 = 0$ at the end of the run by construction.
	In this way it suffices to have only one counter with threshold tests
	which is reused in the subsequent simulations,
	and we obtain decidability by Theorem~\ref{thm:rac:decidable}.
\end{proof}

\end{document}